\documentclass{article}

\usepackage[margin=1.25in,top=1in]{geometry}

\usepackage{graphicx}

\title{Provably faster randomized and quantum algorithms for $k$-means clustering via uniform sampling}
\author{Tyler Chen\thanks{These authors contributed equally. Email: \texttt{\{tyler.chen, archan.ray, akshay.seshadri\}@jpmchase.com}.}\and
Archan Ray\footnotemark[1]\and
Akshay Seshadri\footnotemark[1]\and
Dylan Herman\and 
Bao Bach\and
Pranav Deshpande\and
Abhishek Som\and
Niraj Kumar\thanks{Principal Investigator. Email: \texttt{niraj.x7.kumar@jpmchase.com}.} \and
Marco Pistoia
}
\date{Global Technology Applied Research, JPMorganChase, New York, NY 10001, USA}

\usepackage{amsmath,amssymb,amsthm}
\usepackage{mathtools}
\usepackage{cancel}

\usepackage[labelfont=bf]{caption}
\usepackage{enumitem}
\usepackage{soul}

\usepackage{fourier}

\usepackage[dvipsnames]{xcolor}
\definecolor{c0}{HTML}{1d1d1d}
\definecolor{c1}{HTML}{173359}
\definecolor{c2}{HTML}{8f5a39}
\definecolor{c3}{HTML}{a8332b}
\definecolor{c4}{HTML}{4c825c}
\definecolor{c5}{HTML}{f4efe7}

\usepackage{algorithm,algorithmicx}
\usepackage[noend]{algpseudocode}

\algrenewcommand{\algorithmiccomment}[1]{\hfill \textcolor{gray}{$\vartriangleright$ \textit{#1}}}
\usepackage{etoolbox}

\makeatletter
\expandafter\patchcmd\csname\string\algorithmic\endcsname{\itemsep\z@}{\itemsep=5pt}{}{}
\makeatother

\usepackage{booktabs}
\usepackage{colortbl} 

\usepackage{hyperref}
\hypersetup{
    colorlinks,
    linkcolor=c2,
    citecolor=c2,
    urlcolor=c2
}

\usepackage{thm-restate}
\makeatletter
\renewenvironment{proof}[1][\proofname]{\par
  \pushQED{\qed}%
  \normalfont \topsep6\p@\@plus6\p@\relax
  \trivlist
  \item[\hskip\labelsep
        \it\bfseries
    #1\@addpunct{.}]\ignorespaces
}{%
  \popQED\endtrivlist\@endpefalse
}
\makeatother

\usepackage[nameinlink,capitalize]{cleveref}
\crefformat{equation}{#2(#1)#3}
\crefmultiformat{equation}{{}#2(#1)#3}{ and~#2(#1)#3}{, #2(#1)#3}{ and~#2(#1)#3}

\usepackage[skipabove=.35\parindent,
skipbelow=.35\parindent,
rightline=false,
leftline=true,
topline=false,
bottomline=false,
innerleftmargin=\parindent,
innerrightmargin=\parindent,
innertopmargin=.8\parindent,
innerbottommargin=.6\parindent,
leftmargin=0cm,
rightmargin=0cm,
linewidth=2pt,
backgroundcolor=c5,
linecolor=c2]{mdframed}

\declaretheoremstyle[
    mdframed={} 
]{thm}

\theoremstyle{definition}

\declaretheorem[
    style=thm,
    name=Theorem,
    numberwithin=section
]{theorem}

\declaretheorem[
    style=thm,
    name=Corollary,
    numberlike=theorem
]{corollary}

\global\mdfdefinestyle{alginner}{%
leftmargin=-.2cm,
rightmargin=-.2cm,
innerleftmargin=.2cm,
innerrightmargin=.2cm,
innertopmargin=-.1em,
leftline=false,
backgroundcolor=c5!60,
}

\declaretheorem[
    style=thm,
    name=Algorithm,
    Refname={Algorithm,Algorithms}
]{mdalg}
\usepackage{float}
\floatstyle{plain}

\newmdtheoremenv{lemma}[theorem]{Lemma}
\newmdtheoremenv{imptheorem}[theorem]{Imported Theorem}
\newmdtheoremenv{definition}[theorem]{Definition}
\newmdtheoremenv{remark}[theorem]{Remark}
\newmdtheoremenv{assumption}[theorem]{Assumption}
\newmdtheoremenv{example}[theorem]{Example}

\theoremstyle{definition}

\numberwithin{theorem}{section}
\numberwithin{equation}{section}

\renewcommand{\vec}{\mathbf}
\newcommand{\poly}{\operatorname{poly}}

\newcommand{\T}{{\scalebox{.7}{$\tiny\mathsf{T}$}}}
\newcommand{\F}{{\scalebox{.7}{$\tiny\mathsf{F}$}}}

\newcommand{\EE}{\mathbb{E}}

\newcommand{\argmin}{\operatornamewithlimits{argmin}}

\newcommand{\R}{\mathbb{R}}

\usepackage{dsfont,bm}
\newcommand{\Ind}{\mathds{1}}

\newcommand{\tr}{\operatorname{tr}}

\newcommand{\cjs}{\mathchoice{\big(\vec{c}_1,\ldots,\vec{c}_k\big)}{(\vec{c}_1,\ldots,\vec{c}_k)}{C}{D}}

\newcommand{\cjzs}{\mathchoice{\big(\vec{c}_1^0,\ldots,\vec{c}_k^0\big)}{(\vec{c}_1^0,\ldots,\vec{c}_k^0)}{C}{D}}

\newcommand{\cjhs}{\mathchoice{\big(\widehat{\vec{c}}_1,\ldots,\widehat{\vec{c}}_k\big)}{(\widehat{\vec{c}}_1,\ldots,\widehat{\vec{c}}_k)}{C}{D}}

\usepackage[backend=biber,backref,style=alphabetic,maxalphanames=4,maxcitenames=99,maxbibnames=99,giveninits=true,url=true,doi=false,isbn=false,eprint=false,date=year]{biblatex}
\usepackage{braket}
  \usepackage{nth}
  \usepackage{intcalc}

  \newcommand{\arXiv}[1]{\url{http://arxiv.org/abs/#1}}

\allowdisplaybreaks

\begin{document}
\sloppy
\color{c0}

\maketitle

\begin{abstract}
The $k$-means algorithm (Lloyd's algorithm) is a widely used method for clustering unlabeled data.
A key bottleneck of the $k$-means algorithm is that each iteration requires time linear in the number of data points, which can be expensive in big data applications. 
This was improved in recent works proposing quantum and quantum-inspired classical algorithms to approximate the $k$-means algorithm locally, in time depending only logarithmically on the number of data points (along with data dependent parameters) [q-means: A quantum algorithm for unsupervised machine learning, Kerenidis, Landman, Luongo, and Prakash, NeurIPS 2019; Do you know what $q$-means?, Cornelissen, Doriguello, Luongo, Tang, QTML 2025]. 
In this work, we describe a simple randomized mini-batch $k$-means algorithm and a quantum algorithm inspired by the classical algorithm.
We demonstrate that the worst case guarantees of these algorithms can significantly improve upon the bounds for algorithms in prior work.
Our improvements are due to a careful use of \emph{uniform sampling}, which preserves certain symmetries of the $k$-means problem that are not preserved in previous algorithms that use data norm-based sampling.
\end{abstract}

\section{Introduction}

The $k$-means clustering objective aims to partition a set of data points $\vec{v}_1, \ldots, \vec{v}_n\in\R^d$ into $k$ disjoint sets (also called clusters) such that points within a cluster are close to each other and points in different clusters are far from one another.
Specifically, we are tasked with finding $k$ cluster centers $\vec{c}_1, \ldots, \vec{c}_k\in\R^d$ such that the cost function
\begin{equation}\label{eqn:cost}
\mathcal{L}\cjs
\coloneq \frac{1}{n}\sum_{i\in[n]} \min_{j\in[k]} \| \vec{v}_i -  \vec{c}_j \|^2,
\end{equation}
is minimized. 
Assigning each point $\vec{v}_i$ to the nearest cluster center gives a partition of the data.

Perhaps the most influential algorithm for $k$-means clustering is the $k$-means algorithm (also called Lloyd's algorithm) \cite{lloyd_82}. 
The $k$-means algorithm is a greedy heuristic algorithm\footnote{Finding the optimal partition of $k$ clusters, given the set of $n$ data points is NP-hard \cite{dasgupta_08}.} that takes as input $k$ initial centers (e.g., from $k$-means++ \cite{arthur_vassilvitskii_07}) and iteratively performs two locally optimal steps: (i) the data points are partitioned based on the proximity to the current centroids, and (ii)  new centroids are chosen to optimize the cost function with respect to this partitioning.

We describe \emph{a single iteration} of this algorithm in  \cref{alg:lloyds}.
Typically this procedure is then repeated  (using the output of the previous iteration for the initial centers of the next iteration) for some fixed number of iterations or until a given convergence criterion has been satisfied; e.g., until the centroids or cost do not vary much from iteration to iteration \cite{macqueen_67,lloyd_82}.

\begin{figure}
\begin{mdalg}[Standard $k$-means (one iteration)]~
    \label{alg:lloyds}
    \begin{mdframed}[style=alginner]~
    \begin{algorithmic}[1]
        \Require{Data points $\vec{v}_1, \ldots, \vec{v}_n \in \R^d$, initial centers $\vec c_1^0, \dots, \vec c_k^0 \in \R^d$}
        \For {$j \in [k]$}
        \State $C_j^0 = \Big\{ i\in [n] : j = \argmin_{j'\in[k]} \| \vec{v}_i - \vec{c}_{j'}^0\| \Big\}$
        \label{alg-line:assign-cluster-labels}
        \State $\displaystyle\vec{c}_j = \frac{1}{|C_j^0|}\sum_{i\in C_j^0}  \vec{v}_{i}$
        \EndFor
        \vspace{-1em}\Ensure{Updated centers $\vec{c}_1, \dots, \vec{c}_k \in \R^d$}
    \end{algorithmic}
    \end{mdframed}
\end{mdalg}
\end{figure}

\subsection{Prior Work}
\label{sec:past_quantum}

Each iteration of the $k$-means algorithm (\cref{alg:lloyds}) requires $O(ndk)$ time.
Over recent years quantum and quantum-inspired algorithms have been proposed to reduce the per iteration cost of $k$-means \cite{kerenidis_landman_luongo_prakash_19,doriguello_luongo_tang_25,cornelissen_doriguello_luongo_tang_25}.
Under a balanced clusters assumption (i.e., that $|C_j^0| = \Theta(n/k)$), it is proved that these algorithms output clusters $\widehat{\vec{c}}_1, \ldots, \widehat{\vec{c}}_k$ that are $\varepsilon$-close to the centroids in time $\poly(k,d,\varepsilon^{-1},\log(n))$; see \cref{sec:past_quantum_only} for details.
When the number of data points $n$ is large, such methods offer potential speedups.
Unfortunately, however, the guarantees of past works involve data-dependent quantities such as
\begin{equation}
    \bar{\eta} \coloneq \frac{1}{n}\|\vec{V}\|_\F^2 \hspace{1mm} \text{\cite{doriguello_luongo_tang_25}}
    ,\qquad\hat{\eta} \coloneq \left(\frac{1}{n}\sum_{i\in[n]} \|\vec{v}_i\|\right)^2 + \frac{\|\vec{V}\|^2}{n}\hspace{1mm} \text{\cite{cornelissen_doriguello_luongo_tang_25}}
    ,\qquad 
    \eta \coloneq \|\vec{V}\|_\infty^2 \hspace{1mm}\text{\cite{kerenidis_landman_luongo_prakash_19}}.
\end{equation}
These parameters can be \emph{arbitrarily large}, even on intuitively easy to cluster data-sets.
As we discuss in \cref{sec:symmetry}, the dependence on these parameters is necessary, and arises due to the use of row-norm based importance sampling.
We provide a detailed overview of the access model and guarantees for these algorithms in \cref{sec:past_quantum_only}.

Unsurprisingly, over the past several decades the classical computing community has also developed algorithms to avoid the $O(ndk)$ work per iteration of the $k$-means algorithm.
An important class of these algorithms are based on a technique called \emph{mini-batching}, whereby updates at a given iteration are computed using a small fraction of the total dataset, typically chosen uniformly at random \cite{bottou_bengio_94,so_mahajan_dasgupta_22,tang_monteleoni_17, sculley_10, newling_fleuret_16}.\footnote{The dequantized algorithm of \cite{doriguello_luongo_tang_25} can be viewed as a mini-batch algorithm with non-uniform sampling.}
Such algorithms are widely used in practice, appearing in popular machine learning libraries such as \texttt{scikit-learn} \cite{scikit_11}. A number of these works aim to provide insight into such algorithms by providing theoretical bounds on the batch-size and iteration count required for the algorithm to converge (to e.g., a local minimum) \cite{tang_monteleoni_17,schwartzman_23}.
Often, these mini-batch algorithms make use of additional algorithmic tools such as damping, which helps to reduce variability due to randomness in the algorithms.
However, unlike past work \cite{kerenidis_landman_luongo_prakash_19,doriguello_luongo_tang_25}, such methods do not aim to provide guarantees to mimic the behavior of the $k$-means algorithm over one step. 
In \cref{sec:damping}, we show that one-step guarantees, similar to those proved in \cite{kerenidis_landman_luongo_prakash_19,doriguello_luongo_tang_25}, can be used to derive near-monotonicity bounds for damped mini-batch algorithms, a variant closely related to the {\texttt{MiniBatchKMeans}} method implemented in \texttt{scikit-learn}.

There are a number of other algorithmic techniques that has been explored in the literature to approximate and/or accelerate the convergence of the $k$-means criteria \eqref{eqn:cost}.
The first is an adaptive sampling scheme (called $D^2$ sampling), where at each step, a batch of points are sampled with probability proportional to its squared distance to the nearest center. Notably, $k$-means$++$ \cite{arthur_vassilvitskii_07} uses this scheme to find ``good'' seed points for Lloyd's iterations and has recently been studied in the quantum and quantum-inspired setting \cite{shah_jaiswal_25}.
There is also the large field of literature on coreset construction \cite{phillips_17,bachem_lucic_krause_18,bachem_18sampling,feldman_schmidt_sohler_20}, including in the quantum setting \cite{tomesh_gokhale_21,xue_chen_li_jiang_23}.
In the context of $k$-means, a \emph{coreset} is a (weighted) subset of the data which preserves the cost function \cref{eqn:cost} for \emph{all possible centers}. 
Another related line of work is on dimension reduction, which aims to embed the data in a lower dimensional space while preserving the cost function for all possible centers \cite{cohen_elder_musco_musco_persu_15,boutsidis_zouzias_mahoney_drineas_15}.
Both coreset construction and dimension reduction have primarily been studied in the context of making polynomial time approximation schemes for \cref{eqn:cost} to be computationally tractable.

\subsection{Outline and Contributions}
\label{sec:contributions}

\paragraph{1. Randomized Algorithm with Uniform Sampling:} We analyze a simple mini-batch algorithm based on \emph{uniform sampling} (\cref{alg:minibatch_kmeans}).
The algorithm, about whose novelty we make no claims, is extremely simple: it selects a subset of $b$ data points uniformly at random, and then performs one step of the $k$-means algorithm on this subset.
Our main result is that, if the initial clusters are roughly balanced (i.e., $|C_j^0| = \Theta(n/k)$)\footnote{This assumption enables direct comparison of our guarantees with those in \cite{kerenidis_landman_luongo_prakash_19, doriguello_luongo_tang_25}, which are based on the same assumption.} then the outputs $\widehat{\vec{c}}_1, \ldots, \widehat{\vec{c}}_k$ of \cref{alg:minibatch_kmeans} satisfy, for all $j\in[k]$, $\| \vec{c}_j - \widehat{\vec{c}}_j \|\leq \varepsilon$, if the algorithm uses
\begin{equation}\label{eqn:our_bound_intro}
    O\Bigg( \frac{k^2}{\varepsilon^2} \phi  + \log(k)\Bigg) \text{~samples per iteration},
\end{equation}
where the parameter
\begin{equation}\label{eqn:phi}
    \phi \coloneq \frac{1}{n}\sum_{j\in[k]} \sum_{i\in C_j^0} \| \vec{v}_i - \vec{c}_j \|^2
\end{equation}
measures of the quality of the partition induced by the initial cluster center.

That our bounds depend on $\phi$ is a notable improvement over past work.
Indeed, the true $k$-means step performs well, then we expect $\phi$ to be much smaller than $\bar{\eta}$, $\hat{\eta}$, or $\eta$ (see \cref{thm:leqfrob}).
As we discuss in \cref{sec:symmetry}, a key limitation of the previous works is the unavoidable dependency on quantities like $\eta$, $\hat{\eta}$, and $\bar{\eta}$, which arises from failing to respect a symmetry of the $k$-means objective.

In the appendix, we provide some additional results that may be of independent interest. 
In particular, in \cref{sec:damping} we analyze a generalization of \cref{alg:minibatch_kmeans} which uses damping and is closely related to the implementation of mini-batch $k$-means in \texttt{scikit-learn}.
Our main result shows that, under certain assumptions, the cost \cref{eqn:cost} of the cluster centers produced is nearly monotonic (i.e., does not increase by more than a multiplicity factor $(1+\gamma)$, for some small parameter $\gamma$).

\paragraph{2. Quantum Algorithm:} We develop a quantum algorithm (\cref{alg:quantum_unif_kmeans}) which can be viewed as a quantum analog of \cref{alg:minibatch_kmeans}.
Assuming superposition quantum access to the data using QRAM, and leveraging multivariate quantum mean estimation \cite{cornelissen_amoudi_jerbi_22}, we showcase a quadratic improvement in the $\varepsilon$ dependence compared to our classical algorithm in the small $\varepsilon$ regime.
Specifically, under the same balanced clusters assumption, we obtain the same guarantee using
\begin{equation}\label{eqn:our_bound_intro_quantum}
    \tilde{O}\Bigg( k^{5/2}\sqrt{d}\Bigg(\frac{\sqrt{\phi}}{\varepsilon} + \sqrt{d}\Bigg) \Bigg) \text{~QRAM queries per iteration.}
\end{equation}
As with \cref{alg:minibatch_kmeans}, our guarantees improve on the quantum algorithms from past work \cite{kerenidis_landman_luongo_prakash_19,doriguello_luongo_tang_25}, most notably with regards to the dependence on $\phi$ rather than $\bar{\eta}$, $\hat\eta$, or $\eta$.
In addition, since we only use uniform superposition over the data, we do not require access to a KP-tree data structure.

\paragraph{3. Complexity Comparison:} We compare our quantum algorithm complexity with its classical analogue. 
Under the balanced cluster assumption, the quantum algorithm will do better than the classical algorithm when $\varepsilon$ is smaller than $\tilde{O}(\sqrt{\phi}/\sqrt{d k})$.
Note that when $k = 1$, the clustering problem is equivalent to mean estimation, and it is well-known that we can obtain a quantum advantage only in the small $\varepsilon$ regime~\cite[Theorem~3.7]{cornelissen_amoudi_jerbi_22}.

We also compare our quantum and classical algorithm complexities with previous studies in \cref{table:comparison}. The quantum algorithm in \cite{kerenidis_landman_luongo_prakash_19} does not appear in this table because the authors analyze time complexity, whereas our study focuses on QRAM query complexity.
From \cref{table:comparison}, we can infer that the classical sample complexity bound derived in our study is always better than that of \cite{doriguello_luongo_tang_25}, while at most $\tilde{O}(d)$ worse than the bound derived in \cite{cornelissen_doriguello_luongo_tang_25}. Similarly, for $\varepsilon < 1/\sqrt{d}$, the QRAM query complexity of our quantum algorithm is at most $\tilde{O}(\max\{k, \sqrt{k d}\})$ worse than \cite{doriguello_luongo_tang_25} and at most $\tilde{O}(\max\{\sqrt{d}, \sqrt{k}\} \sqrt{k d})$ worse than the \cite{cornelissen_doriguello_luongo_tang_25}.
On the other hand, $\phi$ can be arbitrarily smaller than both $\bar{\eta}$ and $\hat{\eta}$ (e.g., when the cluster centers are very far; see \cref{ex:hard} and \cref{secn:numerical_validation}), which implies that our bounds can be arbitrarily better than the classical and quantum complexity bounds of \cite{doriguello_luongo_tang_25} and \cite{cornelissen_doriguello_luongo_tang_25}.

\begin{table}[h!]
\centering
\arrayrulecolor{c2}
\begin{tabular}{|c|c|c|c|}
\hline
\rowcolor{c5}
     & \textbf{\cite{doriguello_luongo_tang_25}} & \textbf{\cite{cornelissen_doriguello_luongo_tang_25}} &
     \textbf{Ours} \\
\hline
\textbf{Classical} & $\tilde{O}\left(\frac{k^2 \bar{\eta}}{\varepsilon^2}\right)$ & $\tilde{O}\left(\frac{k^2 \hat{\eta}}{\varepsilon^2}\right)$ & $\tilde{O}\left(\frac{k^2 \phi}{\varepsilon^2}\right)$ \\
\hline
\textbf{Quantum} & $\tilde{O}\left(\frac{k^{3/2} (\sqrt{k} + \sqrt{d}) \sqrt{\bar{\eta}}}{\varepsilon}\right)$ & $\tilde{O}\left(\frac{k^{3/2} \left(\frac{\sqrt{d}}{n}\sum_{i\in[n]} \|\vec{v}_i\| + \frac{\sqrt{k} \|\vec{V}\|}{\sqrt{n}}\right)}{\varepsilon}\right)$ & $\tilde{O}\left(k^{5/2} \sqrt{d} \left(\frac{\sqrt{\phi}}{\varepsilon} + \sqrt{d}\right)\right)$ \\
\hline
\textbf{Comparison} & \multicolumn{3}{c|}{\parbox[m]{10cm}{\begin{center} $\hat{\eta} \leq 2 \bar{\eta}
\leq 2 d \hat{\eta}$ \\ $\phi \leq \bar{\eta} \leq d \hat{\eta}$ \\ $\phi \ll \bar{\eta}, \hat{\eta}$ can hold, e.g., when the cluster centers are far apart\end{center}}} \\
\hline

\end{tabular}
\caption{Comparison of classical sample complexity (Row 1) and QRAM query complexity (Row 2) of the algorithms of \cite{doriguello_luongo_tang_25}, \cite{cornelissen_doriguello_luongo_tang_25}, and our study.  Further details of the QRAM query access model comparison can be found in \cref{remark:access_model_comparison} and \cref{sec:past_quantum_only} respectively. The bounds for the classical and quantum algorithms in \cite{doriguello_luongo_tang_25, cornelissen_doriguello_luongo_tang_25} have been stated in terms of the query access model that is used in our work (1 sample or 1 QRAM query gives the whole row).}
\label{table:comparison}
\end{table}

\subsection{$k$-means problem invariance}\label{sec:symmetry}

The $k$-means clustering objective is invariant to rigid-body transformations of the data (e.g., shifting the data points by some constant vector or rotating the data).
The $k$-means algorithm (\cref{alg:lloyds}) respects this problem symmetry; if run on transformed data (with initial cluster centers transformed accordingly), then the outputs on the transformed problem will be the transformed version of the outputs of the algorithms run on the original data. 
Uniform sampling (and hence \cref{alg:minibatch_kmeans}) also respect this symmetry (and our bounds for the algorithm reflect this).

On the other hand, algorithms which depend on the relative sizes of the row-norms can behave drastically different depending on how the data is transformed (see \cref{fig:probs}).
This invariance is reflected in the analyses of past algorithms, which depend on quantities like $\eta$, $\hat{\eta}$, and $\bar{\eta}$ \cite{kerenidis_landman_luongo_prakash_19,doriguello_luongo_tang_25,cornelissen_doriguello_luongo_tang_25}. 
In realistic scenarios, both $\eta$ and $\bar{\eta}$ can be large, and it is easy to see that a dependence on an aspect-ratio is inherent to algorithms which uses some type of magnitude based importance sampling. 

\begin{figure}[h!]
    \centering
\vspace*{1em}
\includegraphics[scale=.75]{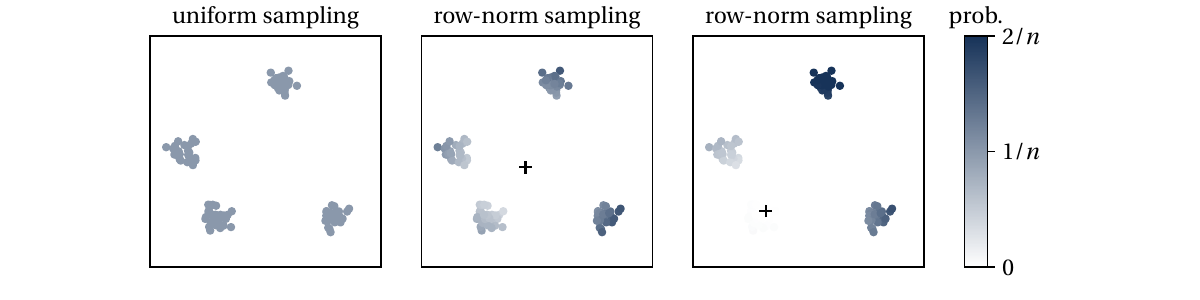}
\caption{Visualization of the probability of sampling a given data point when using uniform sampling or row-norm sampling. 
For row-norm sampling, we show two versions of the data-set, corresponding to a shift (origin is indicated with a plus).
Rigid-body transforms of the data do not change the clustering objective \cref{eqn:cost} or the behavior of the $k$-means algorithm.
However, row-norm sampling is sensitive to  such transforms, and is unlikely to sample points in clusters near the origin.
\vspace*{1.3em}
}
\label{fig:probs}
\end{figure}

\begin{example}\label{ex:hard}
    Consider a dataset $\vec{v}_1, \ldots, \vec{v}_n \in \R^1$ where $\vec{v}_1,\ldots, \vec{v}_{n/2} = -1$ and $\vec{v}_{n/2+1}, \ldots, \vec{v}_n = \alpha \geq 1$.
    This is a perfectly clusterable dataset when $k=2$ (i.e., there are centers for which the cost \cref{eqn:cost} is zero).
    Suppose we initialize with clusters $\vec{c}_1^0 = 0$ and $\vec{c}_2^0 = \alpha$.
\begin{itemize}[leftmargin=1em]
    \item  The $k$-means algorithm (\cref{alg:lloyds}) produces the correct cluster centers $\vec{c}_1 = -1$ and $\vec{c}_2 = \alpha$.

    \pagebreak
    \item The mini-batch $k$-means algorithm (\cref{alg:minibatch_kmeans}) which samples points uniformly produces the correct cluster centers if and only if $B$ contains at least one entry from $\{1, \ldots, n/2\}$ and at least one entry from $\{n/2+1, \ldots, n\}$.
    This happens with probability $1-\delta$ if we use $b = 1+\log_2(1/\delta) = \Theta(\log(1/\delta))$ samples. 

    \item Suppose instead we observe points proportional to their squared row-norms 
    (this is what is done in \cite{doriguello_luongo_tang_25}).
    When $\alpha \gg 1$, we are unlikely to see any points in the cluster at $-1$, in which case we cannot hope to produce an accurate update.
    Specifically,
    \begin{equation*}
    \Pr\Bigg(\text{observe at least one point from cluster at $-1$}\Bigg)=
    1 - \Bigg( 1 - \frac{1}{(1+\alpha^2)} \Bigg)^b,
    \end{equation*}
    so that we need at least $b = \Omega(\alpha^2 \log(1/\delta))$ samples to observe a point from the cluster at $-1$ with probability $1 - \delta$.

    Since $\eta = \max_i \|\vec{v}_i\|^2 = \alpha^2$, $\hat{\eta} = (\frac{1}{n}\sum_{i\in[n]}\|\vec{v}_i\|)^2 = (1 + \alpha)^2/4$, and $\bar{\eta} = \frac{1}{n}\sum_{i\in[n]}\|\vec{v}_i\|^2 = (1 + \alpha^2)/2$, this example demonstrates that a dependence on such quantities is inherent to this type of sampling scheme.

\end{itemize}
\end{example}

\cref{ex:hard} highlights a fundamental weakness of distance-based sampling. 
Intuitively, in order to produce cluster centers similar to those produced by the $k$-means algorithm, we must observe a reasonable number of points from each of the $C_j^0$.
However, even if the clusters are all of similar size (i.e., $|C_j^0| = \Theta(n/k)$) if all the points in $C_j^0$ are near the origin, then we may require a larger number of total samples in order to ensure that we sample enough points in $C_j^0$.
On the other hand, uniform sampling allows to observe a constant fraction samples from each cluster, resulting in a more reliable update.

\subsection{Notation and conventions}

We write $[k] \coloneq \{1,2,\ldots, k\}$, and use $\tilde{O}(\,\cdot\,)$ to hide poly-logarithmic terms in all parameters, including $n$, $\mathcal{L}\cjzs$, $\eta$, etc. 
We use $\| \vec{x} \|$ to indicate the Euclidean norm of a vector $\vec{x}$, $\|\vec{X}\|_\F$ to indicate the Frobenius norm of a matrix $\vec{x}$, and $|C|$ the cardinality of a set $C$.

All of our theoretical guarantees are local per-iteration guarantees, which relate the performance of approximate algorithms to the true $k$-means algorithm. 
Throughout, we will write 
\begin{equation}
    \vec{c}_j = \frac{1}{|C_j^0|}\sum_{i\in C_j^0} \vec{v}_i,
    \qquad
    C_j^0 = \Bigg\{ i\in[n] : j = \argmin_{j'\in[k]} \| \vec{v}_i - \vec{c}_{j'}^0 \| \Bigg\}
\end{equation}
as in the $k$-means algorithm.\footnote{In the case that a data point has multiple nearest initial cluster centers, we assume the tie is broken based on some fixed rule.}
Most of the bounds depend on the parameters
\begin{equation}
    k_C \coloneq  \frac{n}{\min_{j\in[k]}|C_j^0|},
    \qquad 
    \phi \coloneq \frac{1}{n}\sum_{j\in[k]} \sum_{i\in C_j^0} \| \vec{v}_i - \vec{c}_j \|^2,
    \label{eqn:kc_L0}
\end{equation}
which respectively control the relative size of the clusters induced by the initial cluster centers and the quality of these clusters. 
Previous works assume balanced cluster sizes $k_C = \Theta(k)$ \cite{kerenidis_landman_luongo_prakash_19, doriguello_luongo_tang_25}. 

When describing quantum algorithms, $U^\dagger$ will denote the conjugate transpose of a unitary $U$, and $|0\rangle$ will be some easy to prepare state, whose size can be determined from context.

\section{Mini-batch algorithm}
\label{sec:main_results}

We now present \cref{alg:minibatch_kmeans}, which is a mini-batch version of \cref{alg:lloyds}.
The algorithm is simple; in each iteration, a random subset of data points $b$ is drawn from the full dataset uniformly at random (with replacement) and the $k$-means algorithm is applied to this subset.

\begin{figure}
\begin{mdalg}[Mini-batch $k$-means (one iteration)]~
    \label{alg:minibatch_kmeans}
    \begin{mdframed}[style=alginner]~
    \begin{algorithmic}[1]
        \Require{Data $\vec{v}_1, \ldots, \vec{v}_n \in \R^d$, initial centers $\vec c_1^0, \dots, \vec c_k^0 \in \R^d$}
        \State Sample $b$ indices $B = \{s_1, \dots, s_b\} \in [n]^b$, each independently such that $\Pr(s_\ell = i) = 1/n$.
        \For {$j \in [k]$}
        \State $\widehat{C}_j^0 = \Big\{ i\in B : j = \argmin_{j'\in[k]} \| \vec{v}_i - \vec{c}_{j'}^0 \| \Big\}$
        \State $\displaystyle\widehat{\vec{c}}_j =  \frac{1}{|\widehat{C}_j^0|} \sum_{i\in \widehat{C}_j^0} \vec v_{i}$
        \EndFor
        \vspace{-1em}\Ensure{Updated centers $\widehat{\vec{c}}_1, \dots, \widehat{\vec{c}}_k \in \R^d$}
    \end{algorithmic}
    \end{mdframed}
\end{mdalg}
\end{figure}

Strictly speaking, $B$ and $\widehat{C}_j^0$ are multisets (rather than sets), since they may have repeated entries.
We expect an algorithm which draws the entries of $B$ without replacement may perform slightly better. 
However, the algorithm presented is easier to analyze due to the independence of the indices in $B$.

\subsection{Main results}

Our main result for \cref{alg:minibatch_kmeans} is the following.
\begin{restatable}{theorem}{mainthm}\label{thm:main}
Suppose 
\begin{equation*}
    b \geq k_C\cdot \max \Bigg\{ \frac{4\phi}{\varepsilon^{2} \delta}  , 8\log\bigg(\frac{k}{\delta}\bigg) \Bigg\},
\end{equation*}
where $k_C$ and $\phi$ are defined in \cref{eqn:kc_L0}. 
Then, with probability at least $1-\delta$, the output $\cjhs$ of \cref{alg:minibatch_kmeans}
\begin{equation*}
\frac{1}{n}\sum_{j\in[k]} |C_j^0| \| \vec{c}_j - \widehat{\vec{c}}_j \|^2
\leq \varepsilon^2.
\end{equation*}
\end{restatable}

As noted in the introduction, past work \cite{kerenidis_landman_luongo_prakash_19,doriguello_luongo_tang_25} prove that their algorithms mimic, in a certain sense, the $k$-means algorithm. 
While \cref{thm:main} can be converted into a guarantee similar to past work.

\begin{restatable}{corollary}{approxespkmeans}\label{thm:approxepskmeans_cor}
    Suppose
    \begin{equation*}
    b \geq k_C\cdot \max \Bigg\{ \frac{4\phi k_C}{\varepsilon^{2} \delta }  , 8\log\bigg(\frac{k}{\delta}\bigg) \Bigg\},
    \end{equation*}
    where $k_C$ and $\phi$ are defined in \cref{eqn:kc_L0}. 
    Then, with probability at least $1-\delta$, the output $\cjhs$ of \cref{alg:minibatch_kmeans}
    \begin{equation*}
    \forall j\in[k]:\quad \| \vec{c}_j - \widehat{\vec{c}}_j \|
    \leq \varepsilon.
    \end{equation*}
\end{restatable}
In particular, if  $|C_j^0| = \Theta(n/k)$ (i.e., $k_C = \Theta(k)$) and $b = O({k^2}{\varepsilon^{-2}} \phi + \log(k))$, then for all $j\in[k]$, $\| \vec{c}_j - \widehat{\vec{c}}_j \|
\leq \varepsilon$, as mentioned in \cref{eqn:our_bound_intro}.

\begin{remark}\label{rem:failure_prob}
While our bounds have an undesirable polynomial dependence on the inverse failure probability $1/\delta$, the result can be efficiently improved to an arbitrary failure probability by repeating each iteration $O(\log(1/\delta))$ times and using a high-dimensional variant of the ``median-trick''; see \cref{thm:median_of_means}.
\end{remark}

\subsection{Proofs}

In this section we provide our analysis of \cref{alg:minibatch_kmeans}.
Throughout, it will be useful to consider the quantities
\begin{equation}\label{eqn:phi_j}
\phi_j
\coloneq \frac{1}{n}\sum_{i\in C_j^0} \| \vec{v}_i - \vec{c}_j \|^2 ,\qquad j\in[k].
\end{equation}
Note that, by definition (see \cref{eqn:phi}), $\phi = \phi_1 + \cdots + \phi_k$.

\begin{proof}[Proof of \cref{thm:main}.]
Let $\vec{V}^\T = [\vec{v}_1,\ldots, \vec{v}_n] \in \R^{d\times n}$ and, for each $j\in[k]$, define the characteristic vector
\begin{equation}
\vec{x}_j = \frac{1}{|C_j^0|}\sum_{i\in C_j^0} \vec{e}_i ,
\end{equation}
where $\vec{e}_i\in\R^{n}$ is the $i$-th standard basis vector. 
Then $(\vec{x}_j)_i = 1/|C_j^0|$ if $i\in C_j^0$ and $0$ otherwise, and  $\vec{1}^\T\vec{x}_j = 1$, where $\vec{1}$ is the all-ones vector.
As with past work \cite{kerenidis_landman_luongo_prakash_19,doriguello_luongo_tang_25}, we will make use of the fact that $\vec{c}_j = \vec{V}^\T \vec{x}_j$.

Fix $j\in [k]$.
Observe the $j$-th center output by \cref{alg:minibatch_kmeans} can be written as
\begin{equation}
\widehat{\vec{c}}_j = \frac{1}{n|\widehat{C}_j^0|} \sum_{i\in \widehat{C}_j^0} \frac{1}{n^{-1}} \vec{v}_i = \vec{c}_j + \widehat{\lambda}_j \widehat{\vec{y}}_j,
\end{equation}
where
\begin{equation}
\widehat{\lambda}_j = \frac{b|C_j^0|}{n|\widehat{C}_j^0|},
\qquad
\widehat{\vec{y}}_j = \frac{1}{b|C_j^0|} \sum_{i\in \widehat{C}_j^0} \frac{1}{n^{-1}} (\vec{v}_{i} - \vec{c}_j).\footnote{This implies that if $|\widehat{C}_j^0| = 0$, then the algorithm should update $\widehat{\vec{c}}_j = \vec{c}_j$, which would require performing a full $k$-means step. However, this case is avoided with high probability by \cref{eqn:all_lam_small}.}
\end{equation}
Next, note that $\ell_2$-distance between the cluster centers produced by $k$-means and \cref{alg:minibatch_kmeans} can be decomposed as
\begin{equation}\label{cref:bd}
\| \vec{c}_j - \widehat{\vec{c}}_j \|
= \| \widehat{\lambda}_j\widehat{\vec{y}}_j \|
=  |\widehat{\lambda}_j | \|  \widehat{\vec{y}}_j \|.
\end{equation}

We will now show that $|\widehat{\lambda}_j| \leq 2$ with a high probability, and that $\|\widehat{\vec{y}}_j\|$ is small (relative to the per-cluster cost) in expectation. Subsequently employing the linearity of expectation, Markov's inequality, and a union bound gives us the desired bound.

\paragraph{Bound for $|\widehat{\lambda}_j|$:}
Since $|\widehat{C}_j^0|$ is a binomial random variable (with $b$ draws and success probability $|C_j^0|/n$), then $|\widehat{\lambda}_j| = O(1)$ with good probability.
In particular, a standard application of a multiplicative Chernoff bound (see e.g., \cref{thm:chernoff}) gives
\begin{equation}
    \Pr\Bigg( \widehat{\lambda}_j > 2 \Bigg)
    =
    \Pr\Bigg( |\widehat{C}_j^0| < \frac{b|C_j^0|}{2n}\Bigg)
    \leq \exp\Bigg(-\frac{b|C_j^0|}{8n}\Bigg) 
    \leq \exp\Bigg(-\frac{b}{8k_C}\Bigg) 
    \leq  \frac{1}{2\delta k},
\end{equation}
where the last inequality is due to our choice of $b \geq 8 k_C \log(k/\delta)$.
Hence, applying a union bound, we have 
\begin{equation}
\label{eqn:all_lam_small}
    \Pr\Bigg(\forall j\in[k]: \widehat{\lambda}_j \leq 2 \Bigg)
    \geq 1 - \frac{\delta}{2}.
\end{equation}

\paragraph{Bound for $\|\widehat{\vec{y}}_j \|$:}
Observe, again using the definition of $\vec{x}_j$, that $\widehat{\vec{y}}_j$ can be written as
\begin{align}
\widehat{\vec{y}}_j 
&= \frac{1}{b|C_j^0|} \sum_{i\in \widehat{C}_j^0}  \frac{1}{n^{-1}}\vec (\vec{v}_{i} - \vec{c}_j)
\\&= \frac{1}{b|C_j^0|} \sum_{i\in [b]}  \frac{1}{n^{-1}} (\vec{v}_{s_i} - \vec{c}_j) \Ind[s_i\in C_j^0]
\\&= \frac{1}{b} \sum_{i\in [b]}  \frac{1}{n^{-1}}(\vec{v}_{s_i}-\vec{c}_j) (\vec{x}_j)_{s_i}.
\end{align}
A standard direct computation (see \cref{thm:approx_mm}) reveals that
\begin{equation}
\EE\big[\| \widehat{\vec{y}}_j \|^2\big]
= 
\frac{1}{b} \Bigg( \Bigg(\sum_{i\in [n]} \frac{1}{n^{-1}} \|\vec{v}_i - \vec{c}_j\|^2 |(\vec{x}_j)_i|^2\Bigg) 
=\frac{1}{b} \Bigg( \Bigg(\frac{1}{|C_j^0|^2}\sum_{i\in C_j^0} \frac{1}{n^{-1}}\|\vec{v}_i - \vec{c}_j\|^2 \Bigg).
\end{equation}
Therefore, 
\begin{equation}
    \EE\bigg[ |C_j^0|\| \widehat{\vec{y}}_j \|^2 \bigg]
    =\frac{n}{b|C_j^0|} \sum_{i\in C_j^0} \|\vec{v}_i - \vec{c}_j\|^2
    = \frac{n^2\phi_j}{b|C_j^0|}  ,\label{eqn:exact_variance}
\end{equation}
where the last equality is by our definition \cref{eqn:phi_j} of $\phi_j$.
We now use linearity of expectation, the definition $k_C = n / \min_{j\in[k]} |C_j^0|$, and that $\phi = \phi_1+\cdots\phi_k$ to obtain a bound
\begin{equation}
\EE\Bigg[ \frac{1}{n}\sum_{j\in [k]} |C_j^0| \|  \widehat{\vec{y}}_j\|^2 \Bigg]
= \sum_{j\in [k]} \frac{n\phi_j}{b|C_j^0|}
\leq \frac{k_C\phi}{b} .
\end{equation}
Applying Markov's inequality (see \cref{thm:markov}) we find 
\begin{equation}
\Pr\Bigg( \frac{1}{n}\sum_{j\in [k]}|C_j^0| \| \widehat{\vec{y}}_j\|^2 > \frac{\varepsilon^2}{2} \Bigg)
\leq \frac{2k_C\phi}{\varepsilon^2 b}
\leq \frac{\delta}{2},
\end{equation}
where the last inequality is due to our choice of $b \geq 4 k_C \phi/(\delta \varepsilon^2)$.

The result then follows from a union bound with compliment of the event that $|\widehat{\lambda}_j|<2$ for all $j\in [k]$.
\end{proof}

\subsubsection{Proof of $\varepsilon$-approximate $k$-means property}

\cref{thm:approxepskmeans_cor} follows from basic algebraic properties.

\begin{proof}[Proof of \cref{thm:approxepskmeans_cor}]
We have that
\begin{equation}
    \max_{j\in[k]} \| \vec{c}_j - \widehat{\vec{c}}_j \|^2
    = \max_{j\in[k]} \frac{n}{|C_j^0|} \frac{1}{n} |C_j^0|\| \vec{c}_j - \widehat{\vec{c}}_j \|^2
    \leq k_C \frac{1}{n} \max_{j\in[k]} |C_j^0|\| \vec{c}_j - \widehat{\vec{c}}_j \|^2
    \leq k_C \frac{1}{n} \sum_{j\in[k]} |C_j^0|\| \vec{c}_j - \widehat{\vec{c}}_j \|^2.
\end{equation}
The result then follows by \cref{thm:main}, relabeling $\varepsilon$ as appropriate.
\end{proof}

\subsection{Numerical Validation}
\label{secn:numerical_validation}

We now perform numerical experiments to study the behavior of \cref{alg:minibatch_kmeans} as well as the classical algorithms from \cite{doriguello_luongo_tang_25} and \cite{cornelissen_doriguello_luongo_tang_25} (\cref{alg:DLT}), which respectively make use of row-norm squared and row-norm sampling.
Note that these algorithms sample two independent mini-batches, one from a uniform distribution and one proportional to either the squared row-norms or row-norms of the data.
In our experiments we use batch-size $b$ for both the uniform and non-uniform stages (so that $2b$ total samples are used).
Thus, the total number of samples used at a given value of $b$ is twice that of our algorithm.

\paragraph{Batch size:}

The first experiment compares the convergence of \cref{alg:minibatch_kmeans} and \cite[Algorithm 1]{doriguello_luongo_tang_25} as a function of batch-size.
The result is shown in \cref{fig:probs2}.

The data-set is synthetically generated for $d=2$ by choosing $k=4$ means, and then, for each mean, generating $10^4$ data-points from a Gaussian distribution with the given mean (and variance chosen to control $\phi$). Thus $n=4\cdot 10^4$.
This is the same procedure used to generate the data shown in \cref{fig:probs} (but with more points).

We observe that \cref{alg:minibatch_kmeans} clearly outperforms the algorithm of \cite{doriguello_luongo_tang_25}.
This is most pronounced when $\phi$ is small, which is as expected; see \cref{sec:symmetry}.

\begin{figure}[h!]
\centering
\includegraphics[scale=.75]{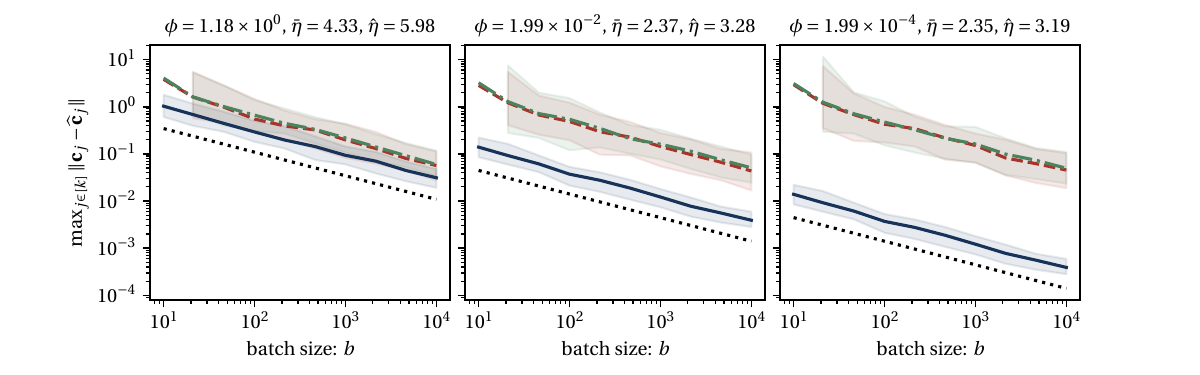}
\caption{
Median and 5\%-95\% error range of cluster error $\max_{j\in[k]} \| \vec{c}_j - \widehat{\vec{c}}_j\|$ over 100 trials for \cref{alg:minibatch_kmeans} (solid),  \cite{doriguello_luongo_tang_25} (dash-dot), and \cite{cornelissen_doriguello_luongo_tang_25} (dash).
Dotted line is the rate $\sqrt{\phi/b}$ predicted by \cref{thm:main}.
}
\label{fig:probs2}
\end{figure}

\paragraph{Multi-step convergence:}

Our theoretical guarantees are only for one step of the $k$-means algorithm.
One might wonder whether \cref{alg:minibatch_kmeans} produces centroids similar to $k$-means when run for multiple iterations.
Our second experiment studies the performance of \cref{alg:minibatch_kmeans} over multiple steps.
The result is shown in \cref{fig:MNIST}.

Here we use the MNIST dataset $n=70,000$, which we embed into $d=30$ dimensions using a Gaussian Johnson--Lindenstrauss embedding. 
This preserves the pairwise distances between points to high relative accuracy.
We then perform clustering with $k=10$ using $k$-means++ initialization \cite{arthur_vassilvitskii_07}.

We observe that the error accumulates some as the number of iterations increases, but remains relatively stable.
Unfortunately, describing this behavior theoretically seems very challenging; since the $k$-means algorithm is non-smooth, any multi-step guarantees would likely require data-dependent assumptions.

\begin{figure}[h!]
\centering
\includegraphics[scale=.75]{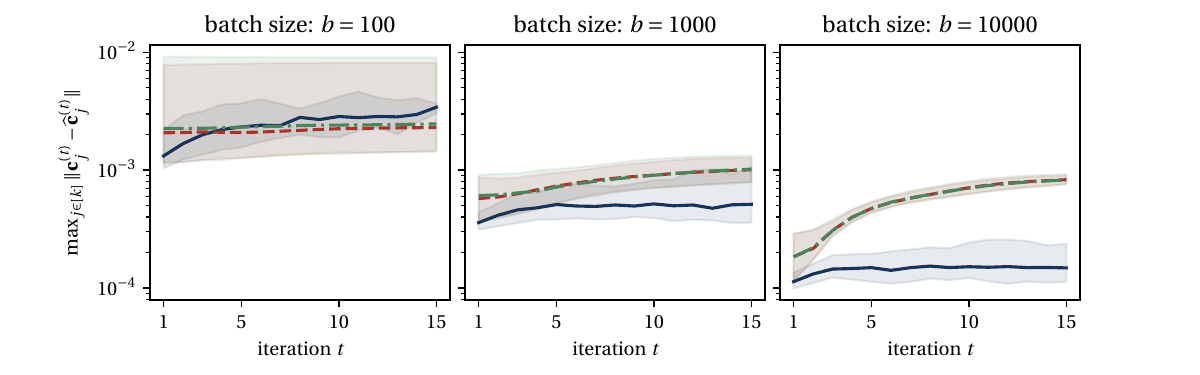}
\caption{
Median and 5\%-95\% error range of cluster error $\max_{j\in[k]} \| \mathbf{c}_j^{(t)} - \widehat{\mathbf{c}}_j^{(t)} \|$ at iteration $t$ over 100 trials for \cref{alg:minibatch_kmeans} (solid),  \cite{doriguello_luongo_tang_25} (dash-dot), and \cite{cornelissen_doriguello_luongo_tang_25} (dash).
}
\label{fig:MNIST}
\end{figure}

\section{Quantum algorithm}
\label{sec:quantum_alg}

Changing the computational access model to the data from the classical random-access-memory (RAM) to the quantum superposition access model via QRAM allows for performing coherent operations on the quantum states, often resulting in quadratic reduction in the QRAM queries compared to the classical queries \cite{grover_96, brassard_hoyer_mosca_tapp_00}. In this section, simplifying the results by \cite{kerenidis_landman_luongo_prakash_19, doriguello_luongo_tang_25}, we propose an improved quantum algorithm for $k$- means clustering which returns a $\varepsilon$ accurate approximation to the exact $d$-dimensional centroids per-iteration using  using $\tilde{O}(\sqrt{d}/\varepsilon)$ queries (as opposed to $\tilde{O}(1/\varepsilon^2)$ queries required by classical algorithms).

\subsection{Computational model}

A key assumption of our quantum algorithm is that the input data can accessed in a coherent superposition using a quantum computer.
\begin{assumption}[Quantum query access]\label{asm:qram}
We assume input data $\vec{x}_1, \ldots, \vec{x}_q\in\R^d$ is available through a quantum data structure (which we call QRAM) that provides access to the unitary $U$ implementing the map $|i\rangle|0\rangle\mapsto |i\rangle | \vec{x}_{i} \rangle$.
A ``query to the QRAM'' refers to the use of this map (or its inverse).
\end{assumption}

This is similar to the RAM access model assumed for the classical algorithm, $i\mapsto \vec{x}_i$ (and as with the classical algorithm, we only read entire columns of the data-matrix at once).
The main difference is that the unitary which implements $|i\rangle|0\rangle\mapsto |i\rangle | \vec{x}_{i} \rangle$ also allows us to act simultaneously on states in superposition
\begin{equation}
    \sum_{i} \alpha_i \ket{i}\ket{0}\mapsto \sum_{i}\alpha_i \ket{i}\ket{\vec{x}_i}
\end{equation}
for arbitrary $\alpha_i \in \mathbb{C}$ such that $\sum_i |\alpha_i|^2 = 1$.

Analogously to our analysis of the classical algorithm in \cref{sec:main_results}, \emph{we measure the ``cost'' of our quantum algorithm in terms of the number of queries to the QRAM.}\footnote{There are several proposals for data structures which efficiently implement QRAM \cite{giovannetti_lloyd_maccone_08, kerenidis_prakash_17}, and hence it would be possible to measure the cost in terms of something like the gate complexity.
However, quantum query models (such as the one used here) are popular because they allow natural and simple analyses of many core quantum algorithms, as well as fine-grained lower-bounds \cite{buhrman_dewolf_02,aaronson_21}.
}

There are a number of ways to encode a real-number $x$ into the quantum state $|x\rangle$ \cite{nielsen_chuang_10}. 
The precise encoding is not important for us, so long as we can perform certain basic operations, which are used within our algorithm and subroutines such as quantum mean estimation \cite{cornelissen_amoudi_jerbi_22}.

\begin{assumption}[Exact quantum arithmetic]\label{asm:arithmetic}
We assume real-numbers $x$ can be represented in a quantum state $|x\rangle$ and that we can perform basic arithmetic/comparison operations on such quantum states exactly; e.g., $|x\rangle|0\rangle\mapsto |x\rangle|\sqrt{x}\rangle$, $|x,y\rangle|0\rangle \mapsto |x,y\rangle|x+y\rangle$, $|x,y\rangle|0\rangle \mapsto |x,y\rangle|x\cdot y\rangle$,
$|x,y\rangle|0\rangle \mapsto |x,y\rangle|\Ind[x<y]\rangle$, etc.
\end{assumption}

This is similar to the assumption we make in the analysis of our classical algorithm (i.e., that arithmetic is exact).
For instance, the quantum state $|x\rangle$ can be thought of as the \emph{bit-encoding} of $x$:
\begin{equation}
|x\rangle = |x_1\rangle |x_2\rangle \cdots |x_b\rangle,
\end{equation}
where $x_i\in\{0,1\}$ is the $i$-th bit of a $b$-bit binary finite-precision representation of $x$.
Classical fixed and floating-point arithmetic circuits have analogs which can be implemented on quantum computers.
Hence, \cref{asm:arithmetic} essentially states that we use enough bits of precision that the impacts of finite precision are negligible compared to other errors.
More efficient circuits designed specifically for quantum computing are also a topic of study \cite{thomsen_gluck_axelsen_10,ruizperez_garciaescartin_17,wang_li_lee_deb_lim_chattopadhyay_25}

\subsection{Core primitives}

The first main tool we use is the fixed-point amplitude amplification algorithm, which boosts the amplitude of a ``good state'' to a large probability with a quadratic improvement in the complexity over classical analogue \cite{brassard_hoyer_mosca_tapp_00,nielsen_chuang_10,gilyen_su_low_wiebe_19}.

\begin{imptheorem}[\protect{Fixed-point amplitude amplification~\cite[Theorem~27]{gilyen_su_low_wiebe_19}}]\label{thm:amplitude_amp}
    Assume access to a unitary $Q$ which performs the map
    \begin{equation*}
    |0,0\rangle\mapsto \sum_{\ell \in [n]} \alpha_\ell |\psi_\ell\rangle|\ell\rangle.
    \end{equation*}
    Fix $j \in [n]$.
    Given parameters $\Delta \in (0, 1)$ and $\alpha \in(0,\alpha_j)$, there is a quantum algorithm that uses $O(\alpha^{-1} \log(1/\Delta))$ queries to $Q$ and $Q^\dagger$ to construct a unitary that performs the map 
    \begin{equation*}
    |0,0\rangle \mapsto \sqrt{1-\Delta'}|\psi_j\rangle|j\rangle + \sqrt{\Delta'}|G\rangle,
    \end{equation*}
    for some $\Delta' \leq \Delta$, where $|G\rangle$ is orthogonal to $|\psi_j\rangle|j\rangle$.
    The state $|\psi_j\rangle$ need not be known.
\end{imptheorem}
\begin{proof}
    Let $\ket{\psi_0} = \ket{0, 0}$ and $\ket{\psi} = Q \ket{\psi_0}$.
    For any orthogonal projector $P$, define $C_{P}\operatorname{NOT} = X \otimes P + I \otimes (I - P)$, where $X$ is the Pauli $X$ matrix.
    Let $\Pi_j = I \otimes \ket{j} \bra{j}$, and $P_0 = \ket{\psi_0} \bra{\psi_0}$.
    Observe that $\Pi_j \ket{\psi} = \alpha_j \ket{\psi_j} \ket{j}$.
    Then, by~\cite[Theorem~27]{gilyen_su_low_wiebe_19}, there is a unitary $\tilde{Q}$ that can be constructed using $O(\alpha^{-1} \log(1/
    \Delta)))$ queries to $Q$, $Q^\dagger$, $C_{\Pi_j}\operatorname{NOT}$, $C_{P_0}\operatorname{NOT}$, and $e^{i \varphi Z}$ (where $Z$ is the Pauli $Z$ matrix), such that
    \begin{equation}
        \|\ket{\psi_j} \ket{j} - \tilde{Q} \ket{\psi_0}\| \leq \sqrt{\Delta}. 
        \label{eqn:amp_dist_Delta}
    \end{equation}
    Decompose $\tilde{Q} \ket{\psi_0} = \sqrt{1-\Delta'} \ket{\psi_j} \ket{j} + \sqrt{\Delta'} \ket{G}$, where $\ket{G}$ orthogonal to $\ket{\psi_j} \ket{j}$.
    Then
    $\Delta \geq \|\ket{\psi_j} \ket{j} - \tilde{Q} \ket{\psi_0}\|^2 = \|(1-\sqrt{1-\Delta'})\ket{\psi_j} \ket{j} - \sqrt{\Delta'} \ket{G}\|^2 = (1-\sqrt{1-\Delta'})^2 + \Delta' \geq \Delta'$.
\end{proof}

The second important tool we use is a quantum mean estimation algorithm.\footnote{While \cite[Theorem 3.5]{cornelissen_amoudi_jerbi_22} gives an improvement over classical estimators for small $\varepsilon$, no improvement over classical estimators is possible for large $\varepsilon$ \cite[Theorems 3.7]{cornelissen_amoudi_jerbi_22}.}

\begin{imptheorem}[\protect{Quantum mean estimation \cite[Theorem 3.5]{cornelissen_amoudi_jerbi_22}}]\label{thm:quantum_mean_est}
    Let $\Omega$ be a finite sample space and $\vec{X}$ a $d$-dimensional random variable taking value $\vec{X}(\omega)$ with probability $p(\omega)$. 
    Define the mean and covariance
    \begin{equation*}
    \bm{\mu} = \sum_{\omega\in\Omega} p(\omega) \vec{X}(\omega)
    \qquad
    \bm{\Sigma} = \sum_{\omega\in\Omega} p(\omega) (\vec{X}(\omega) - \bm{\mu})(\vec{X}(\omega) - \bm{\mu})^\T.
    \end{equation*}
    Assume access to unitaries $U$ and $B$ that implement the maps
    \begin{equation*}
    |0,0\rangle \mapsto \sum_{\omega\in\Omega} \sqrt{p(\omega)} | \omega \rangle|G_\omega\rangle \qquad
    |\omega\rangle |0\rangle \mapsto |\omega\rangle|\vec{X}(\omega)\rangle,
    \end{equation*}
    respectively.
    Here $|G_\omega\rangle$ can be any state.\footnotemark

    Fix $\delta>0$. For $\varepsilon  < \log(d/\delta)/\sqrt{d}$, there is a quantum algorithm that uses $\tilde{O}( {\sqrt{d}} / {\varepsilon} )$ queries to $U$ and $B$ to output an estimate $\widehat{\bm{\mu}}$ of $\bm{\mu}$ such that, with probability at least $1-\delta$,
    \begin{equation*}
    \| \bm{\mu} - \widehat{\bm{\mu}} \|^2
    \leq \varepsilon^2 \tr(\bm{\Sigma}).
    \end{equation*}
\end{imptheorem}
\footnotetext{The probability $p(\omega)$ is recovered by the projective measurement $\| \Pi_\omega U|0,0\rangle\|^2$, where $\Pi_\omega = |\omega\rangle\langle \omega|\otimes I$ are the orthogonal projectors corresponding to a measurement (in the standard basis) of the first register.}

\subsection{Algorithm and main result}\label{sec:quantumalgorithm}

\cref{alg:quantum_unif_kmeans} describes the quantum algorithm for performing one-step of the $k$-means algorithm.
Our main result theoretical result is the following.

\begin{figure}
\begin{mdalg}[Quantum uniform $k$-means (one iteration)]~
    \label{alg:quantum_unif_kmeans}
    \begin{mdframed}[style=alginner]~
    \begin{algorithmic}[1]
        \Require{Data $\vec{v}_1, \ldots, \vec{v}_n \in \R^d$, initial centers $\vec c_1^0, \dots, \vec c_k^0 \in \R^d$ (both with quantum query access)}
        \For{$j=1,\ldots, k$}
        \State 
        Use \cref{alg:boosting} to construct unitary $\widetilde{U}_j$ (approximately) implementing the map:
        \begin{equation*}
        |0,0\rangle\mapsto  \sum_{i\in C_j^0} \frac{1}{\sqrt{|C_j^0|}} |i \rangle|j\rangle.
        \end{equation*}
        \State\label{alg-line:Bj} Use \cref{alg:Bj} to construct unitary $B_j$ implementing the map:
        \begin{equation*}
            |i\rangle|\vec{0}\rangle \mapsto |i\rangle|\vec{X}_j(i)\rangle,
        \end{equation*}
        \Statex\hspace{\algorithmicindent} where
        \begin{equation*}
            \vec{X}_j(i) = \begin{cases}
                \vec{v}_i - \vec{c}_j^0 & \ell_i=j \\
                \vec{0}&\text{otherwise}
            \end{cases},
        \end{equation*}
        \Statex\hspace{\algorithmicindent} and $\ell_i = \argmin_{j'\in[k]} \{\| \vec{v}_i - \vec{c}_{j'}^0 \|\}$.
        \State 
        Use mean estimation (\cref{thm:quantum_mean_est}) with $\widetilde{U}_j$ and $B_j$ to obtain estimate $\widehat{\bm{\mu}}_j$. \label{alg:quantum_unif_kmeans:muj}
        \State Set $\widehat{\vec{c}}_j = \vec{c}_{j}^0 + \widehat{\bm{\mu}}_j$.
        \EndFor
        \vspace{-1em}\Ensure{Updated centers $\widehat{\vec{c}}_1, \dots, \widehat{\vec{c}}_k \in \R^d$}
    \end{algorithmic}
    \end{mdframed}
\end{mdalg}
\end{figure}

\begin{theorem}
\label{thm:quantum_main}
Given QRAM access (\cref{asm:qram}) to data points $\vec{v}_1, \ldots, \vec{v}_n$ and cluster centers $\vec{c}_1^0, \ldots, \vec{c}_k^0$ such that the quantum data encoding satisfies \cref{asm:arithmetic},
Then, for $\varepsilon > 0$, using 
\begin{equation*}
    \tilde{O}\Bigg( k^{3/2} k_C \sqrt{d}\ \Bigg(\frac{\sqrt{\phi}}{\varepsilon} + \sqrt{\frac{dk}{k_C}}\Bigg)\Bigg) \quad \textnormal{QRAM queries},
\end{equation*}
with probability at least $1-\delta$, the output $\cjhs$ of \cref{alg:quantum_unif_kmeans} satisfies
\begin{equation*}
\forall j\in[k]:\quad \| \vec{c}_j - \widehat{\vec{c}}_j \| \leq \varepsilon,
\end{equation*}
where $k_C$ is defined in \cref{eqn:kc_L0}.
\end{theorem}
In particular, assuming the balanced cluster size assumption $k_C = \Theta(k)$ \cite{kerenidis_landman_luongo_prakash_19, doriguello_luongo_tang_25}, we get the bound \cref{eqn:our_bound_intro_quantum} from the introduction.
We note that the $\tilde{O}(\,\cdot\,)$ in \cref{thm:quantum_main} hides a logarithmic factors in $n$, $d$, $\delta$, and $\mathcal{L}\cjzs$.\footnote{$\mathcal{L}\cjzs$ is incomparable to the quantities terms $\bar{\eta}$ and $\eta$ appearing in past work.
However, for any reasonable initialization (i.e., better than $\vec{c}_j^0 = \vec{0}$), we expect $\mathcal{L}\cjzs$ to be small relative to these quantities.}
The exact dependencies, which are suppressed in the theorem statement, can be seen in the proof.
We have not tried to optimize the dependence on $k$, which we suspect can be improved.

\begin{remark}
    As noted in \cref{sec:contributions}, the precise access model and way of measuring costs for the quantum algorithms in \cite{kerenidis_landman_luongo_prakash_19}, \cite{doriguello_luongo_tang_25, cornelissen_doriguello_luongo_tang_25}, and our work are all slightly different.
    In particular, \cite{kerenidis_landman_luongo_prakash_19} assumes the the data is prepared in quantum states using amplitude encoding.
    This is a weaker assumption than we make in the current paper, since a bit-encoding can be efficiently converted to an amplitude encoding.
    \cite{doriguello_luongo_tang_25, cornelissen_doriguello_luongo_tang_25} use bit-encoding, similar to the present paper, and additionally assumes that the QRAM can efficiently prepare row-norm weighted superpositions over the data. 
    While the time-complexity is reported, as noted in \cref{foot:DLT25}, it seems that the analysis omits the time cost of mean estimation and rather reports its QRAM query complexity.
    We also note that \cite{doriguello_luongo_tang_25, cornelissen_doriguello_luongo_tang_25} count the cost of querying individual entries of the data (as opposed to querying a whole row as we do in our study).
    \label{remark:access_model_comparison}
\end{remark}

\subsection{Subroutines of the algorithm}

We now describe the key subroutines of \cref{alg:quantum_unif_kmeans}.

Note that the centers output by the $k$-means algorithm \cref{alg:lloyds} are simply the mean of the points within each cluster $C_j^0$.
Therefore, in order to perform a similar mean estimation step in \cref{alg:quantum_unif_kmeans} over a given cluster $j\in[k]$, we would like access to a unitary $U_j$ which performs the map
\begin{equation}\label{eqn:within_cluster_uniform}
    |0\rangle \mapsto \sum_{i\in C_j^0} \frac{1}{\sqrt{|C_j^0|}} |i\rangle,
\end{equation}
and additionally access to a unitary $B_j$ implementing the map
\begin{equation*}
    |i\rangle|\vec{0}\rangle \mapsto |i\rangle|\vec{X}_j(i)\rangle,
\end{equation*}
 where
        \begin{equation*}
            \vec{X}_j(i) = \begin{cases}
                \vec{v}_i - \vec{c}_j^0 & \ell_i=j \\
                \vec{0}&\text{otherwise}
            \end{cases},
        \end{equation*}
and $\ell_i = \argmin_{j'\in[k]} \{\| \vec{v}_i - \vec{c}_{j'}^0 \|\}$.

\subsubsection{Cluster assignment}

As with past work \cite{kerenidis_landman_luongo_prakash_19,doriguello_luongo_tang_25}, the starting point of algorithm is implementing a unitary which provides labels to the data points, which is described in \Cref{alg:quantum_cluster_assignment}.

\Cref{alg:quantum_cluster_assignment} simply uses the QRAM to load $\vec{v}_i$ and $\vec{c}_1^0, \ldots, \vec{c}_k^0$, and then applies reversible analogs of classical operations to identify the nearest cluster.

\begin{theorem}[Cluster assignment]\label{thm:quantum_cluster_assignment}
\Cref{alg:quantum_cluster_assignment} uses $O(k)$ QRAM queries to construct a unitary $U$ which performs the map:
\begin{equation*}
     \forall i\in[n]: |i,0\rangle  \mapsto |i,\ell_i\rangle,
    \qquad \ell_i = \argmin_{j\in[k]} \| \vec{v}_i - \vec{c}_j^0 \|.
\end{equation*}
\end{theorem}

\begin{proof}
In \cref{alg:quantum_cluster_assignment:QRAM}, we use $k+1$ calls to the QRAM.
The rest of the operations are implemented by quantum analogs of classical circuits, and we uncompute the ancilla registers afterwards.
Note that the arithmetic operations and minimum finding are agnostic to the value of $i$ because they act on whatever data is loaded by the QRAM.
Since each line of \cref{alg:quantum_cluster_assignment} is a unitary operation and the product of unitary maps is unitary, the resulting map that takes $\ket{i} \ket{0}$ to $\ket{i} \ket{\ell_i}$ for all $i \in [n]$ is unitary.
\end{proof}

We note that it may be possible to get a better $k$ dependence by using a quantum minimization algorithm that can approximately compute the argmin of $k$ numbers using $O(\sqrt{k})$ QRAM queries, compared to an exact classical algorithm that uses $O(k)$ queries.
Such an approach was used in \cite{doriguello_luongo_tang_25}, although the impact of the error was not analyzed.
We leave such an analysis to future work.

\begin{mdalg}[Quantum cluster assignment]~
    \label{alg:quantum_cluster_assignment}
    \begin{mdframed}[style=alginner]~
    \begin{algorithmic}[1]
        \Require{Data $\vec{v}_1, \ldots, \vec{v}_n \in \R^d$, initial centers $\vec c_1^0, \dots, \vec c_k^0 \in \R^d$ (both with quantum query access)}
        \State Use $k+1$ calls to the QRAM to load the data: \label{alg:quantum_cluster_assignment:QRAM}
        \begin{equation*}
            |i\rangle|\vec{0},\ldots, \vec{0}\rangle
            \mapsto |i\rangle |\vec{v}_i\rangle |\vec{c}_1^0, \ldots, \vec{c}_k^0\rangle.
        \end{equation*}
        \State Use quantum arithmetic to perform:
        \begin{equation*}
            |i\rangle |\vec{v}_i\rangle |\vec{c}_1^0, \ldots, \vec{c}_k^0\rangle |0, \dotsc, 0\rangle
            \mapsto |i\rangle |\|\vec{c}_1^0-\vec{v}_i\|^2, \ldots, \|\vec{c}_k^0-\vec{v}_i\|^2\rangle.
        \end{equation*}
        \State Use classical minimization circuit to perform:
        \begin{equation*}
            |i\rangle |\|\vec{c}_1^0-\vec{v}_i\|^2, \ldots, \|\vec{c}_k^0-\vec{v}_i\|^2 \rangle |0\rangle
            \mapsto |i\rangle |\ell_i\rangle,
        \end{equation*}
        where $\ell_i = \argmin_{j'\in[k]} \{\| \vec{v}_i - \vec{c}_{j'}^0 \|\}$.
        \Ensure Unitary to perform: 
        \begin{equation*}
            \forall i\in[n]:\quad |i\rangle |0\rangle \mapsto |i\rangle |\ell_i\rangle.
        \end{equation*}
    \end{algorithmic}
    \end{mdframed}
\end{mdalg}

\subsubsection{Per-cluster superposition}

In \cref{alg:boosting} we describe how to approximately implement $U_j$.
The key observation is that a uniform superposition of the labeled data can be used to obtain an approximation to $U_j$.
Indeed,
\begin{equation}
\label{eqn:superposition_aa}
\sum_{i\in[n]}\frac{1}{\sqrt{n}} |i,\ell_i\rangle
= \sum_{j\in[k]} \sum_{i\in C_j^0} \frac{1}{\sqrt{n}} |i,j\rangle
= \sum_{j\in[k]} \sqrt{\frac{|C_j^0|}{n}} \Bigg(\sum_{i\in C_j^0} \frac{1}{\sqrt{|C_j^0|}} |i\rangle \Bigg) |j\rangle.
\end{equation}
Since we know the ``good state'' has the form $|\psi_j\rangle|j\rangle$, we can use fixed point amplitude amplification guarantee \cref{thm:amplitude_amp} to implement a unitary performing an approximate version of the map \cref{eqn:within_cluster_uniform}.

\begin{corollary}\label{thm:boosting}
    Fix $\Delta \in (0, 1)$. 
    Assume access to a unitary $U$ that performs 
    \begin{equation*}
         \forall i\in[n]: |i,0\rangle  \mapsto |i,\ell_i\rangle,
        \qquad \text{where }\ell_i = \argmin_{j\in[k]} \| \vec{v}_i - \vec{c}_j^0 \|.
    \end{equation*}
    For all $j\in[k]$, \cref{alg:boosting} uses at most $O(\sqrt{k_C} \log(1/\Delta))$ queries to $U$ and $U^\dagger$ to construct a unitary $\widetilde{U}_j$ which performs the map:
    \begin{equation*}
        |0\rangle|0\rangle \mapsto \sqrt{1-\Delta'}\sum_{i\in C_j^0} \frac{1}{\sqrt{|C_j^0|}} |i\rangle|j\rangle + \sqrt{\Delta'} |G\rangle,
    \end{equation*}
    for some $\Delta'\leq\Delta$ and garbage state $|G\rangle$ orthogonal to the desired state.
\end{corollary}

\begin{proof}
The uniform superposition over $i \in [n]$ can be created using a unitary map~\cite{shukla_vedula_24}.
Next, note that the state created after the application of the unitary $U$ constructed by \cref{alg:quantum_cluster_assignment} on this uniform superposition can be expressed as
\begin{equation}
    \sum_{i\in[n]} \frac{1}{\sqrt{n}} |i\rangle |\ell_i\rangle = \sum_{j \in [k]} \sqrt{\frac{|C_j^0|}{n}} \sum_{i \in C_j^0} \frac{1}{\sqrt{|C_j^0|}} |i\rangle |j\rangle.
\end{equation}
Thus, there is a unitary $U'$ that creates the above state from $\ket{0, 0}$ using a single call to $U$.
Then, amplitude amplification (see \cref{thm:amplitude_amp}) of the state $\ket{j}$ (in the second register) with $O(\sqrt{n/|C_j^0|} \log(1/\Delta)) \leq O(\sqrt{k_C} \log(1/\Delta))$ queries to $U'$, $(U')^\dagger$ gives a unitary $\tilde{U}_j$ that acts on $|0, 0\rangle$ to produce
\begin{equation}
    \sqrt{1-\Delta'} \sum_{i\in C_j^0}\frac{1}{\sqrt{|C_j^0|}} |i\rangle |j\rangle + \sqrt{\Delta'} |G\rangle,
\end{equation}
where $\ket{G}$ is some state orthogonal to $\sum_{i\in C_j^0} |i\rangle |j\rangle$.
As a result, we need at most $O(\sqrt{k_C} \log(1/\Delta))$ calls to $U$ and $U^\dagger$ to construct the map $\widetilde{U}_j$.
Note that we uncompute the ancilla registers.
\end{proof}

\begin{mdalg}[Approximate sampling unitary preparation]~
    \label{alg:boosting}
    \begin{mdframed}[style=alginner,]~
    \begin{algorithmic}[1]
        \Require{Index $j \in [k]$, parameter $\Delta \in (0, 1)$}
        \State Create the uniform superposition:
        \begin{equation*}
            |0\rangle \mapsto \sum_{i\in[n]} \frac{1}{\sqrt{n}} |i\rangle .
        \end{equation*}
    \State Use \cref{alg:quantum_cluster_assignment} to construct a unitary $U$ that performs:
        \begin{equation*}
            \sum_{i\in[n]} \frac{1}{\sqrt{n}} |i\rangle |0\rangle \mapsto \sum_{i\in[n]} \frac{1}{\sqrt{n}} |i\rangle |\ell_i\rangle,
        \end{equation*}
        where $\ell_i = \argmin_{j'\in[k]} \{\| \vec{v}_i - \vec{c}_{j'}^0 \|\}$.    
        \State Perform amplitude amplification  (see \cref{thm:amplitude_amp})  with $O\big(\sqrt{{n/|C_j^0|}} \log(1/\Delta)\big)$ queries to $U$ and $U^\dagger$ to obtain unitary $V_j$ which performs the map
        \[
        \sum_{i\in[n]} \frac{1}{\sqrt{n}} |i\rangle |\ell_i\rangle \mapsto \sqrt{1-\Delta'} \sum_{i\in C_j^0}\frac{1}{\sqrt{|C_j^0|}} |i\rangle |j\rangle + \sqrt{\Delta'} |G\rangle
        \]
        for some $\Delta' \leq \Delta$ and garbage state $|G\rangle$ orthogonal to the desired state.
        \Ensure Unitary $\widetilde{U}_j$ to (approximately) perform: 
        \begin{equation*}
            |0\rangle |0\rangle \mapsto \sum_{i\in C_j^0} \frac{1}{\sqrt{|C_j^0|}} |i\rangle |j\rangle.
        \end{equation*}
    \end{algorithmic}
    \end{mdframed}
\end{mdalg}

\subsubsection{Random variable access}

In \cref{alg:Bj}, we describe how to implement the unitary $B_j$ that allows us to query the random variable $\vec{X}_j$ in \cref{alg:quantum_unif_kmeans}~\cref{alg-line:Bj}.
Here, we observe that using QRAM to $|i,\ell_i\rangle$ allows one to load $|\vec{v}_i\rangle|\vec{c}_{\ell_i}^0\rangle$.
Subsequently, we can use exact quantum arithmetic \cref{asm:arithmetic} to implement the target unitary.

\begin{corollary}
    \label{thm:Bj}
    Assume access to a unitary $U$ that performs 
    \begin{equation*}
        \forall i \in [n]: |i,0\rangle  \mapsto |i,\ell_i\rangle,
        \qquad \text{where }\ell_i = \argmin_{j\in[k]} \| \vec{v}_i - \vec{c}_j^0 \|.
    \end{equation*}
    For each $j \in [k]$, \cref{alg:Bj} uses $O(k)$ calls to the QRAM to prepare the unitary $B_j$.
\end{corollary}

\begin{proof}
    By \cref{thm:quantum_cluster_assignment}, \cref{alg:Bj} uses $O(k)$ queries to the QRAM to construct a unitary that maps $\ket{i} \ket{0}$ to $\ket{i} \ket{\ell_i}$ for all $i \in [n]$.
    It follows that we make a total of $O(k)$ QRAM queries in \cref{alg:Bj}.
    Note that the arithmetic operations are agnostic to the value of $i$ because they act on whatever data is loaded by the QRAM.
    Since all the lines in \cref{alg:Bj} are unitary operations, and the composition of unitary maps is unitary, the map $B_j$ is unitary.
    Note that we uncompute the ancilla registers at the end of the computation.
\end{proof}

\begin{mdalg}[Random variable access unitary preparation]~
    \label{alg:Bj}
    \begin{mdframed}[style=alginner]~
    \begin{algorithmic}[1]
        \Require{Data $\vec{v}_1, \ldots, \vec{v}_n \in \R^d$ and initial centers $\vec c_1^0, \dots, \vec c_k^0 \in \R^d$ (with quantum query access), index $j \in [k]$}
        \State Use \cref{alg:quantum_cluster_assignment} to perform the map:
        \begin{equation*}
            |i\rangle |0\rangle \mapsto |i\rangle |\ell_i\rangle,
        \end{equation*}
        where $\ell_i = \argmin_{j'\in[k]} \{\| \vec{v}_i - \vec{c}_{j'}^0 \|\}$.
        \State Use $2$ calls to the QRAM to load the data:
        \begin{equation*}
            |i\rangle |\ell_i\rangle |\vec{0}\rangle |\vec{0}\rangle \mapsto |i\rangle |\ell_i\rangle |\vec{v}_i\rangle |\vec{c}_{\ell_i}^0\rangle.
        \end{equation*}
        \State Use quantum arithmetic to perform:
        \begin{equation*}
            |i\rangle |\ell_i\rangle |\vec{v}_i\rangle |\vec{c}_{\ell_i}^0\rangle |\vec{0}\rangle \mapsto |i\rangle |\ell_i\rangle |\vec{v}_i\rangle |\vec{c}_{\ell_i}^0\rangle |(\vec{v}_i - \vec{c}_{\ell_i}^0) \Ind[\ell_i = j]\rangle.
        \end{equation*}
        \Ensure Unitary $B_j$ to perform:
        \begin{equation*}
            \forall i \in [n]:\quad|i\rangle |\vec{0}\rangle \mapsto |i\rangle |(\vec{v}_i - \vec{c}_{\ell_i}^0) \Ind[\ell_i = j]\rangle.
        \end{equation*}
    \end{algorithmic}
    \end{mdframed}
\end{mdalg}

\subsection{Proof of main quantum bound}

We now provide the proof of \cref{thm:quantum_main} which 
essentially amounts to applying the quantum mean estimation guarantee \cref{thm:quantum_mean_est} while controlling the impact of using $\widetilde{U}_j$ in place of $U_j$.

Throughout, it will be useful to consider the per-cluster cost
\begin{equation}\label{eqn:per_cluster_cost}
\mathcal{L}_j\cjzs
\coloneq \frac{1}{n}\sum_{i\in C_j^0} \| \vec{v}_i - \vec{c}_j^0 \|^2.
\end{equation}
Since $C_1^0, \ldots, C_k^0$ forms a partition of $[n]$, we have that $\mathcal{L}_1\cjzs +\cdots + \mathcal{L}_k\cjzs = \mathcal{L}\cjzs$.

\begin{proof}[Proof of \cref{thm:quantum_main}.]
Fix $j\in[k]$.
We will use the sample space
\begin{equation}\label{eqn:sample_space}
    \Omega = [n]
\end{equation}
and consider the random variable 
\begin{equation}
\vec{X}_j(i) = 
\begin{cases}
    \vec{\vec{v}}_i - \vec{c}_{j}^0 & \ell_i=j \\
    \vec{0} & \text{otherwise}
\end{cases},
\end{equation}
which is encoded by the unitary $B_j$ of \cref{alg:quantum_unif_kmeans}~\cref{alg-line:Bj}.
Recall here that $\ell_i = \argmin_{j'\in[k]} \{\| \vec{v}_i - \vec{c}_{j'}\|\}$.
As noted in \cref{thm:Bj}, \cref{alg:Bj} used to construct $B_j$, requires $O(k)$ QRAM queries.

Before proceeding, it is informative to consider what would happen if we had a unitary $U_j$ performing \cref{eqn:within_cluster_uniform}.
Let $\Pi_i = |i\rangle\langle i|\otimes I$ be the orthogonal projector onto the subspace spanned by $|i\rangle$ in the first register.
Applying $U_j \otimes I$ to $\ket{0, 0}$ and performing the measurement $\{\Pi_i\}$ gives the probability distribution
\begin{equation}
\label{eqn:Uj_probs}
p_j(i) 
= \| \Pi_i  (U_j \otimes I) |0, 0\rangle\|^2
= \begin{cases}
    |C_{j}^0|^{-1} & \ell_i = j \\
    0 & \text{otherwise}
\end{cases},
\qquad
i\in [n]
\end{equation}
on the sample space \cref{eqn:sample_space}.
Mean estimation (\cref{thm:quantum_mean_est}) with $U_j$ and $B_j$ therefore produces an approximation of
\begin{equation}
\bm{\mu}_j 
= \EE[\vec{X}_j]
= \sum_{i} p_j(i) \vec{X}_j(i)
=  \sum_{i\in C_j^0} \frac{1}{|C_j^0|} (\vec{v}_i - \vec{c}_j^0)
= \vec{c}_j - \vec{c}_j^0,\label{eqn:mu-j}
\end{equation}
from which we could obtain an approximation to $\vec{c}_j$.
The corresponding variance term is 
\begin{align}
\tr(\bm{\Sigma}_j)
= \EE\big[\|\vec{X}_j - \bm{\mu}_j\|^2\big]
&= \sum_{i\in[n]} p_j(i)\| \vec{X}_j(i) - (\vec{c}_j -\vec{c}_j^0) \|^2
\\&= \sum_{i\in C_j^0} \frac{1}{|C_j^0|} \| \vec{v}_i - \vec{c}_j \|^2
\\&= \frac{n}{|C_j^0|} \phi_j,
\end{align}
where the last equality is by our definition of the per-cluster cost \cref{eqn:phi_j}.

Of course, we do not have access to $U_j$. 
Instead, we have access to $\widetilde{U}_j$, which corresponds to a perturbed version of the distribution described in \cref{eqn:Uj_probs}, which we will write as $\widetilde{p}_j(i)$.
As guaranteed by \cref{thm:boosting}, \cref{alg:boosting} constructs $\widetilde{U}_j$ using $O(\sqrt{k_C} \log(1/\Delta))$ queries to
the unitary $U$ (implemented by \cref{alg:quantum_cluster_assignment}) and $U^\dagger$ that performs $|0,0\rangle \mapsto \sum_{i\in [n]}\frac{1}{\sqrt{n}}|i, \ell_i\rangle$, where $\ell_i = \argmin_{j\in [k]}\|\vec{v}_i - \vec{c}_j^0\|$.
Since we need $O(k)$ QRAM queries to construct $U$ (and also $U^\dagger$), as shown in \cref{thm:quantum_cluster_assignment}, we need a total of
\begin{equation}
    O\left( k \sqrt{k_C} \log\left( \frac{1}{\Delta} \right) \right) \text{ QRAM queries} \label{eqn:boosting_qram_queries}
\end{equation}
to construct $\widetilde{U}_j$.
Here, $\Delta$ is an accuracy parameter, to be determined, which controls the nearness of $\widetilde{p}_j(i)$ to $p_j(i)$.

From \cref{thm:boosting}, we know that $\widetilde{U}_j|0,0\rangle = \sqrt{1-\Delta'} (U_j \otimes I) |0,0\rangle + \sqrt{\Delta'}|G\rangle$ for some $\Delta'\leq\Delta$, and therefore,
\begin{equation}
    \widetilde{p}_j(i) 
    = \|\Pi_i \widetilde{U}_j |0, 0\rangle\|^2
    = \big\|\sqrt{1-\Delta'} \Pi_i (U_j \otimes I) |0, 0\rangle + \sqrt{\Delta'} \Pi_i |G\rangle\big\|^2
    ,\qquad
    i\in [n].
\end{equation}
Thus, since $\|\Pi_i (U_j \otimes I) |0, 0\rangle\| \leq 1$ and  $\| \Pi_i |G\rangle \| \leq 1$, we have
\begin{align}
    \forall i\in [n]: \quad|p_j(i) - \widetilde{p}_j(i)|
    &= \bigg|\bigg(\sqrt{p_j(i)} + \sqrt{\widetilde{p}_j(i)}\bigg)\bigg(\sqrt{p_j(i)} - \sqrt{\widetilde{p}_j(i)}\bigg)\bigg|
    \\&\leq \bigg(\sqrt{p_j(i)} + \sqrt{\widetilde{p}_j(i)}\bigg)
            \big\|(\sqrt{1-\Delta'} - 1) \Pi_i (U_j \otimes I) |0, 0\rangle + \sqrt{\Delta'} \Pi_i |G\rangle\big\| \\
    &\leq 4 \sqrt{\Delta'} \\
    &\leq 4 \sqrt{\Delta},\label{eqn:probability_bound}
    \end{align}
where we used reverse triangle inequality in the second inequality.

Now, define $\widetilde{\bm{\mu}}_j = \widetilde{\EE}[\vec{X}_j] = \sum_{i \in [n]} \widetilde{p}_j(i) \vec{X}_j(i)$ to be the expected value of $\vec{X}_j$ with respect to the probabilities induced by $\widetilde{U}_j$.
Let $\widehat{\bm{\mu}}_j$ denote the output of mean estimation in \cref{alg:quantum_unif_kmeans:muj} of \cref{alg:quantum_unif_kmeans}.
By the triangle inequality and definition of $\widehat{\vec{c}}_j$ in \cref{alg:quantum_unif_kmeans}, 
\begin{equation}\label{eqn:quantum_triangle}
    \| \vec{c}_j - \widehat{\vec{c}}_j \|
    = \| (\vec{c}_j - \vec{c}_j^0) -  \widehat{\bm{\mu}}_j \|
    \leq \| \widetilde{\bm{\mu}}_j - (\vec{c}_j - \vec{c}_j^0) \| + \| \widetilde{\bm{\mu}}_j -  \widehat{\bm{\mu}}_j \|.
\end{equation}
We will now bound each of these terms by $2\varepsilon_j \cdot \sqrt{k_C \smash[b]{\phi_j}}$, with $\varepsilon_j$ to be determined. 
Subsequently, we choose $\varepsilon_j$ such that the sum of errors is bounded above by $\varepsilon$.

With the benefit of foresight, we set
\begin{equation}
    \varepsilon_j = \min\Bigg\{\frac{\varepsilon}{3 \sqrt{k_C \smash[b]{\phi_j}}}, \frac{\log(kd/\delta)}{\sqrt{2 d}}\Bigg\}, \qquad
    \Delta_j = \frac{\min\big\{1,\varepsilon_j^2\big\}}{16 n^2 k_C^2} \frac{(\phi_j)^2}{\mathcal{L}_j\cjzs^2}, 
    \label{eqn:varepsilonj}
\end{equation}
and
\begin{equation}    
    \Delta = \min_{j \in [k]} \Delta_j.
    \label{eqn:Delta_choice}
\end{equation}

\paragraph{Bound for $\| \widetilde{\bm{\mu}}_j - (\vec{c}_j - \vec{c}_j^0) \|$:}

From \eqref{eqn:mu-j}, we know $\EE[\vec{X}_j] = \vec{c}_j - \vec{c}_j^0$. Since the distributions $p_j$ and $\widetilde{p}_j$ are close, we expect $\widetilde{\EE}[\vec{X}_j]$ to be close to $\EE[\vec{X}_j]$.

Note that $\vec{X}_j(i) = \vec{v}_i - \vec{c}_j^0$ if $i \in C_j^0$ and $\vec{0}$ otherwise.
Therefore, we have
\begin{equation}
    \widetilde{\bm{\mu}}_j = \sum_{i\in [n]} (p_j(i) + \widetilde{p}_j(i) - p_j(i)) \vec{X}_j(i)
    = (\vec{c}_j - \vec{c}_j^0) + \sum_{i \in C_j^0}(\widetilde{p}_j(i,j) - p_j(i,j)) (\vec{v}_i - \vec{c}_j^0).
\end{equation}
Let $\vec{A}$ be the $d \times |C_j^0|$ size matrix with columns given by $\vec{v}_i - \vec{c}_j^0$ for $i \in C_j^0$, and $\vec{b}$ to be the $|C_j^0|$-dimension vector with entries given by $\widetilde{p}_j(i) - p_j(i)$ for $i \in C_j^0$.
Then, we can write
\begin{equation}
    \sum_{i \in |C_j^0|} (\widetilde{p}_j(i) - p_j(i)) (\vec{v}_i - \vec{c}_j^0)
    = \vec{A} \vec{b}.
\end{equation}
Then, using submultiplicativity, $\|\vec{A} \vec{b}\| \leq \|\vec{A}\|_\F \|\vec{b}\|$, we obtain
\begin{align}
    \|  \widetilde{\bm{\mu}}_j - (\vec{c}_j - \vec{c}_j^0)\|^2
    &= \bigg \| \sum_{i\in C_j^0} (\widetilde{p}_j(i,j) - p_j(i,j)) (\vec{v}_i - \vec{c}_j^0) \bigg \|^2
    \\&\leq  \Bigg(\sum_{i\in C_j^0} |\widetilde{p}_j(i,j) - p_j(i,j)|^2 \Bigg) \Bigg(\sum_{i\in C_j^0}  \|\vec{v}_i - \vec{c}_j^0 \|^2\Bigg)
    \\&\leq 16 n^2 \Delta \mathcal{L}_j\cjzs \\
    &\leq \varepsilon_j^2 k_C \phi_j,
    \label{eqn:tildemuj_cj}
\end{align}
where in the second-last step we used \cref{eqn:probability_bound} and the definition of per-cluster cost given in \cref{eqn:per_cluster_cost}.
The last step follows from the choice of $\Delta$ in \cref{eqn:Delta_choice}, $k_C \geq 1$, and the fact that $\phi_j/\mathcal{L}_j\cjzs \leq 1$.

\paragraph{Bound for $\| \widetilde{\bm{\mu}}_j -  \widehat{\bm{\mu}}_j \|$:}

Recall that $\widehat{\bm{\mu}}_j$ is the output of mean estimator for $\vec{X}_j$ with respect to the $\widetilde{p}_j(i)$ probabilities.
Let $\widetilde{\bm{\Sigma}}_j$ be the covariance of the random variable $\vec{X}_j$ with respect to these probabilities.

The choice of $\varepsilon_j$ in \cref{eqn:varepsilonj} guarantees $\varepsilon_j  < \log(kd/\delta)$, and hence, by \cref{thm:quantum_mean_est},
\begin{equation}
\label{eqn:tildemuj}
    \| \widetilde{\bm{\mu}}_j -  \widehat{\bm{\mu}}_j \|
    \leq \varepsilon_j \sqrt{\tr(\widetilde{\bm{\Sigma}}_j)},
\end{equation}
with probability at least $1-\delta/k$, using $\tilde{O}(\sqrt{d}/\varepsilon_j)$ queries to $\widetilde{U}_j$ and $B_j$, and therefore by \cref{eqn:boosting_qram_queries} and \cref{thm:Bj},
\begin{equation}
    \tilde{O}\Bigg( \frac{k \sqrt{d k_C}}{\varepsilon_j} \log\Bigg( \frac{1}{\Delta} \Bigg) \Bigg)
    = \tilde{O}\Bigg( k k_C \sqrt{d} \log\Bigg( \frac{1}{\Delta} \Bigg) \max\Bigg\{\frac{\sqrt{\phi_j}}{\varepsilon}, \sqrt{\frac{d}{k_C}} \frac{1}{\log(kd/\delta)}\Bigg\} \Bigg) \text{ QRAM queries}.
    \label{eqn:qram_percluster_cost}
\end{equation}

We now bound $\tr(\widetilde{\bm{\Sigma}}_j)$.
Using that $\tr(\widetilde{\bm{\Sigma}}_j) = \widetilde{\EE}[\|\vec{X}_j - \widetilde{\bm{\mu}}_j \|^2] \leq \widetilde{\EE}[\|\vec{X}_j - \vec{c} \|^2]$ for all vectors $\vec{c}$, we obtain
\begin{equation}
    \tr(\widetilde{\bm{\Sigma}}_j)
    = \widetilde{\EE}\big[\|\vec{X}_j - \widetilde{\bm{\mu}}_j \|^2\big] 
    \leq \widetilde{\EE}\big[\|\vec{X}_j - \bm{\mu}_j \|^2\big]
    = \sum_{i\in C_j^0} \widetilde{p}_j(i) \| \vec{X}_j(i) -  \bm{\mu}_j \|^2
    + \sum_{i\in[n]\setminus C_j^0} \widetilde{p}_j(i) \| \vec{X}_j(i) -  \bm{\mu}_j \|^2.
    \label{eqn:sigmatilde_bd}
\end{equation}
The first term on the right-hand-side of \cref{eqn:sigmatilde_bd} is bounded as
\begin{align}
    \sum_{i\in C_j^0} \widetilde{p}_j(i) \| \vec{X}_j(i) -  \bm{\mu}_j \|^2
    &= \sum_{i\in C_j^0} \widetilde{p}_j(i) \| (\vec{v}_i - \vec{c}_j^0) - (\vec{c}_j - \vec{c}_j^0) \|^2
    \\&\leq \sum_{i\in C_j^0} \Bigg(\frac{1}{|C_j^0|} + 4 \sqrt{\Delta}\Bigg) \| \vec{v}_i - \vec{c}_j \|^2
    \\&= \Bigg(\frac{n}{|C_j^0|} + 4 n \sqrt{\Delta}\Bigg) \phi_j,    \label{eqn:sigmatilde_bd:1}
\end{align}
where the inequality follows from the bound \cref{eqn:probability_bound}.

The second term on the right-hand-side of \cref{eqn:sigmatilde_bd} is bounded by
\begin{align}
    \sum_{i\in[n]\setminus C_j^0} \widetilde{p}_j(i) \| \vec{X}_j(i) -  \bm{\mu}_j \|^2
    & = \sum_{i\in[n]\setminus C_j^0} \widetilde{p}_j(i) \| \vec{0} -  (\vec{c}_j - \vec{c}_j^0) \|^2
    \\&\leq \sum_{i\in[n]\setminus C_j^0} 4 \sqrt{\Delta} \|\vec{c}_j - \vec{c}_j^0 \|^2
    \\&\leq 4 n \sqrt{\Delta} \|\vec{c}_j - \vec{c}_j^0 \|^2,
    \label{eqn:sigmatilde_bd:2}
\end{align}
where the inequality also follows from the bound \cref{eqn:probability_bound}.
Finally, observe that by a ``bias-variance'' type identity,
\begin{equation}
    n \mathcal{L}_j\cjzs = \sum_{i\in C_j^0} \| \vec{v}_i - \vec{c}_j^0 \|^2
    =  \sum_{i\in C_j^0} \| \vec{v}_i - \vec{c}_j \|^2
    + |C_j^0| \| \vec{c}_j - \vec{c}_j^0 \|^2
    \geq |C_j^0| \| \vec{c}_j - \vec{c}_j^0 \|^2.
    \label{eqn:sigmatilde_bd:3}
\end{equation}

Plugging \cref{eqn:sigmatilde_bd:3} into \cref{eqn:sigmatilde_bd:2}, and then plugging \cref{eqn:sigmatilde_bd:1,eqn:sigmatilde_bd:2} into \cref{eqn:sigmatilde_bd}, we get the bound
\begin{equation}
\begin{aligned}
    \tr(\widetilde{\bm{\Sigma}}_j)
    &\leq
    \Bigg(\frac{n}{|C_j^0|} + 4 n \sqrt{\Delta}\Bigg) \phi_j
    + \frac{n}{|C_j^0|} 4 n \sqrt{\Delta} \mathcal{L}_j\cjzs \\
    &\leq k_C \phi_j + 8 n k_C \sqrt{\Delta} \mathcal{L}_j\cjzs \\
    &\leq 4 k_C \phi_j,
\end{aligned}
\end{equation}
where we used the fact that $1 \leq n/|C_j^0| \leq k_C$ and $\phi_j \leq \mathcal{L}_j\cjzs$ in the second step, and the choice of $\Delta$ in \cref{eqn:Delta_choice} in the last step.
Therefore, \cref{eqn:tildemuj} becomes
\begin{equation}
    \| \widetilde{\bm{\mu}}_j -  \widehat{\bm{\mu}}_j \|
    \leq 2 \sqrt{k_C \smash[b]{\phi_j}} \varepsilon_j.
\end{equation}
Using this with \cref{eqn:tildemuj_cj} and \cref{eqn:quantum_triangle}, we obtain
\begin{equation}
    \|\vec{c}_j - \widehat{\vec{c}}_j\| \leq 3 \sqrt{k_C \smash[b]{\phi_j}} \varepsilon_j \leq \varepsilon
\end{equation}
with probability at least $\delta/k$, where the last inequality follows from the definition of $\varepsilon_j$ in \cref{eqn:varepsilonj}.

\paragraph{Putting it together:}
Finally, using the union bound, we obtain $\|\vec{c}_j - \widehat{\vec{c}}_j\| \leq \varepsilon$ simultaneously for all $j \in [k]$ with probability at least $1-\delta$.
We sum \cref{eqn:qram_percluster_cost} over $j\in [k]$ to get a total cost of
\begin{equation}
    \tilde{O}\Bigg( k^{3/2} k_C\sqrt{d} \log\Bigg(\frac{1}{\Delta}\Bigg) \Bigg(\frac{\sqrt{\phi}}{\varepsilon} + \sqrt{\frac{d k}{k_C}} \frac{1}{\log(kd/\delta)}\Bigg) \Bigg) \text{ QRAM queries},
\end{equation}
where we use the fact that $\max\{x,y\}\leq x+y$ for non-negative $x,y$, and that $\sum_{i = 1}^k \sqrt{\phi_i} \leq \sqrt{k\sum_{i = 1}^k \phi_i} = \sqrt{k \phi}$ by concavity of square-root.
Then, using $\log(kd/\delta) \geq 1$ completes the proof.
\end{proof}

\section{Outlook}

We have described and analyzed classical (randomized) and quantum algorithms for the $k$-means problem. 
Our analysis shows that these algorithms produce cluster centers near to those produced by the classical $k$-means algorithm (over one iteration) with query complexity bounds depending on the parameter $\phi$ \cref{eqn:phi}. This a notable improvement over past work \cite{kerenidis_landman_luongo_prakash_19,doriguello_luongo_tang_25} which incur dependency on $\eta = \max_i \|\vec{v}_i\|^2$ and $\bar{\eta} = n^{-1} \sum_i \|\vec{v}_i\|^2$ respectively; indeed $\phi \leq \bar{\eta}\leq \eta$ always holds. 
This improvement is due to a more fine-grained analysis of the algorithms described in this paper and a key observation that the algorithms proposed by \cite{kerenidis_landman_luongo_prakash_19,doriguello_luongo_tang_25} are sensitive to rigid-body transforms of the data, and can therefore fail to produce good centers even on intuitively easy to cluster datasets. The algorithms we study, whose core is \emph{uniform sampling}, are invariant to such transforms leading to the improved bounds.

There are a number of directions for future work. 
For instance, it would be interesting to understand the behavior of the algorithms studied in this paper over multiple iterations. 
The primary difficulty limiting this extension is that the $k$-means algorithm is not smooth; a tiny change to the initialization can produce very different centers, even after one step. 
However, for certain types of problems, it may be possible to derive guarantees on the centers or the cost.
It would also be interesting to derive quantum algorithms for other clustering objectives, including on spaces other than $\R^d$ (e.g., on the space of probability distributions). 

More broadly, there are a number of related algorithmic techniques which might be able to accelerate algorithms for one-step guarantees.
The first is the large field of literature on coreset construction \cite{phillips_17,bachem_lucic_krause_18,bachem_18sampling,feldman_schmidt_sohler_20}, including in the quantum setting \cite{tomesh_gokhale_21,xue_chen_li_jiang_23}.
In the context of $k$-means, a \emph{coreset} is a (weighted) subset of the data which preserves the cost function \cref{eqn:cost} for \emph{all possible centers}. 
Another related line of work is on dimension reduction, which aims to embed the data in a lower dimensional space while preserving the cost function for all possible centers \cite{cohen_elder_musco_musco_persu_15,boutsidis_zouzias_mahoney_drineas_15}.
Both coreset construction and dimension reduction have primarily been studied in the context of making polynomial time approximation schemes for \cref{eqn:cost} more computationally tractable.

\section*{Acknowledgements}
We thank our colleagues at the Global Technology Applied Research center of JPMorganChase for support
and helpful feedback. Special thanks to Brandon Augustino, Shouvanik Chakrabarti, Jacob Watkins, and
Jamie Heredge for their valuable discussions regarding the manuscript.

\section*{Disclaimer}

This paper was prepared for informational purposes by the Global Technology Applied Research center of JPMorgan Chase \& Co. This paper is not a merchandisable/sellable product of the Research Department of JPMorgan Chase \& Co. or its affiliates. Neither JPMorgan Chase \& Co. nor any of its affiliates makes any explicit or implied representation or warranty and none of them accept any liability in connection with this paper, including, without limitation, with respect to the completeness, accuracy, or reliability of the information contained herein and the potential legal, compliance, tax, or accounting effects thereof. This document is not intended as investment research or investment advice, or as a recommendation, offer, or solicitation for the purchase or sale of any security, financial instrument, financial product or service, or to be used in any way for evaluating the merits of participating in any transaction.

\printbibliography

@inproceedings{shah_jaiswal_25,
  author       = {Poojan Chetan Shah and
                  Ragesh Jaiswal},
  title        = {Quantum (Inspired) D2-sampling with Applications},
  booktitle    = {The Thirteenth International Conference on Learning Representations,
                  {ICLR} 2025, Singapore, April 24-28, 2025},
  publisher    = {OpenReview.net},
  year         = {2025},
  url          = {https://openreview.net/forum?id=tDIL7UXmSS},
  bibsource    = {dblp computer science bibliography, https://dblp.org}
}

@inproceedings{arthur_vassilvitskii_07,
author = {Arthur, David and Vassilvitskii, Sergei},
title = {\texttt{k-means++}: the advantages of careful seeding},
year = {2007},
isbn = {9780898716245},
publisher = {Society for Industrial and Applied Mathematics},
address = {USA},
booktitle = {Proceedings of the Eighteenth Annual ACM-SIAM Symposium on Discrete Algorithms},
pages = {1027–1035},
numpages = {9},
location = {New Orleans, Louisiana},
series = {SODA '07}
}

@misc{brassard_hoyer_mosca_tapp_00,
   title={Quantum amplitude amplification and estimation},
   ISSN={0271-4132},
   url={http://dx.doi.org/10.1090/conm/305/05215},
   DOI={10.1090/conm/305/05215},
   journal={Quantum Computation and Information},
   publisher={American Mathematical Society},
   author={Brassard, Gilles and Høyer, Peter and Mosca, Michele and Tapp, Alain},
   year={2002},
   pages={53–74}
}

@article{bachem_lucic_krause_18,
  series = {KDD ’18},
  title = {Scalable k -Means Clustering via Lightweight Coresets},
  url = {http://dx.doi.org/10.1145/3219819.3219973},
  DOI = {10.1145/3219819.3219973},
  journal = {Proceedings of the 24th ACM SIGKDD International Conference on Knowledge Discovery \& Data Mining},
  publisher = {ACM},
  author = {Bachem,  Olivier and Lucic,  Mario and Krause,  Andreas},
  year = {2018},
  month = {jul},
  pages = {1119–1127},
  collection = {KDD ’18}
}

@phdthesis{bachem_18sampling,
  title={Sampling for large-scale clustering},
  author={Bachem, Olivier Fr{\'e}d{\'e}ric},
  year={2018},
  school={ETH Zurich}
}

@inproceedings{bottou_bengio_94,
    author = {Bottou, L\'{e}on and Bengio, Yoshua},
    booktitle = {Advances in Neural Information Processing Systems},
    editor = {G. Tesauro and D. Touretzky and T. Leen},
    pages = {},
    publisher = {MIT Press},
    title = {Convergence Properties of the K-Means Algorithms},
    url = {https://proceedings.neurips.cc/paper_files/paper/1994/file/a1140a3d0df1c81e24ae954d935e8926-Paper.pdf},
    volume = {7},
    year = {1994}
}

@article{boutsidis_zouzias_mahoney_drineas_15,
    title = {Randomized Dimensionality Reduction for $k$-Means Clustering},
    volume = {61},
    ISSN = {1557-9654},
    url = {http://dx.doi.org/10.1109/TIT.2014.2375327},
    DOI = {10.1109/tit.2014.2375327},
    number = {2},
    journal = {IEEE Transactions on Information Theory},
    publisher = {Institute of Electrical and Electronics Engineers (IEEE)},
    author = {Boutsidis,  Christos and Zouzias,  Anastasios and Mahoney,  Michael W. and Drineas,  Petros},
    year = {2015},
    month = {feb},
    pages = {1045–1062}
}

@inproceedings{cohen_elder_musco_musco_persu_15,
    series = {STOC ’15},
    title = {Dimensionality Reduction for k-Means Clustering and Low Rank Approximation},
    url = {http://dx.doi.org/10.1145/2746539.2746569},
    DOI = {10.1145/2746539.2746569},
    booktitle = {Proceedings of the forty-seventh annual ACM symposium on Theory of Computing},
    publisher = {ACM},
    author = {Cohen,  Michael B. and Elder,  Sam and Musco,  Cameron and Musco,  Christopher and Persu,  Madalina},
    year = {2015},
    month = {jun},
    pages = {163–172},
    collection = {STOC ’15}
}

@inproceedings{cornelissen_amoudi_jerbi_22,
    series = {STOC ’22},
    title = {Near-optimal Quantum algorithms for multivariate mean estimation},
    url = {http://dx.doi.org/10.1145/3519935.3520045},
    DOI = {10.1145/3519935.3520045},
    booktitle = {Proceedings of the 54th Annual ACM SIGACT Symposium on Theory of Computing},
    publisher = {ACM},
    author = {Cornelissen,  Arjan and Hamoudi,  Yassine and Jerbi,  Sofiene},
    year = {2022},
    month = {jun},
    collection = {STOC ’22}
}

@misc{dasgupta_08,
    title={The hardness of k-means clustering},
    author={Dasgupta, Sanjoy},
    url={https://escholarship.org/uc/item/2qm3k10c},
    journal={UC San Diego: Department of Computer Science & Engineering},
    year={2008}
}

@misc{doriguello_luongo_tang_25,
      title={Do you know what $q$-means?}, 
      author={Joao F. Doriguello and Alessandro Luongo and Ewin Tang},
      year={2025},
      eprint={2308.09701v2},
      archivePrefix={arXiv},
      primaryClass={quant-ph},
      url={https://arxiv.org/abs/2308.09701v2}, 
}

@inproceedings{cornelissen_doriguello_luongo_tang_25,
      title={Do you know what $q$-means?}, 
      author={Arjan Cornelissen and Joao F. Doriguello and Alessandro Luongo and Ewin Tang},
      year={2025},
    booktitle = {Quantum Techniques in Machine Learning (QTML 2025)},
      eprint={2308.09701v3},
      archivePrefix={arXiv},
      primaryClass={quant-ph},
      url={https://arxiv.org/abs/2308.09701v3}, 
}

@article{feldman_schmidt_sohler_20,
    title = {Turning Big Data Into Tiny Data: Constant-Size Coresets for  $k$-Means,  PCA,  and Projective Clustering},
    volume = {49},
    ISSN = {1095-7111},
    url = {http://dx.doi.org/10.1137/18M1209854},
    DOI = {10.1137/18m1209854},
    number = {3},
    journal = {SIAM Journal on Computing},
    publisher = {Society for Industrial & Applied Mathematics (SIAM)},
    author = {Feldman,  Dan and Schmidt,  Melanie and Sohler,  Christian},
    year = {2020},
    month = {jan},
    pages = {601–657}
}

@inproceedings{kerenidis_prakash_17,
  doi = {10.4230/LIPICS.ITCS.2017.49},
  url = {https://drops.dagstuhl.de/entities/document/10.4230/LIPIcs.ITCS.2017.49},
  author = {Kerenidis,  Iordanis and Prakash,  Anupam},
  language = {en},
  title = {Quantum Recommendation Systems},
  publisher = {Schloss Dagstuhl – Leibniz-Zentrum f\"{u}r Informatik},
  year = {2017},
  copyright = {Creative Commons Attribution 3.0 Unported license}
}

@inbook{kerenidis_landman_luongo_prakash_19,
    author = {Kerenidis, Iordanis and Landman, Jonas and Luongo, Alessandro and Prakash, Anupam},
    title = {q-means: a quantum algorithm for unsupervised machine learning},
    year = {2019},
    publisher = {Curran Associates Inc.},
    address = {Red Hook, NY, USA},
    booktitle = {Proceedings of the 33rd International Conference on Neural Information Processing Systems},
    articleno = {372},
    numpages = {11},
    url={https://proceedings.neurips.cc/paper/2019/hash/16026d60ff9b54410b3435b403afd226-Abstract.html}
}

@article{lloyd_82,
    title = {Least squares quantization in PCM},
    volume = {28},
    ISSN = {0018-9448},
    url = {http://dx.doi.org/10.1109/TIT.1982.1056489},
    DOI = {10.1109/tit.1982.1056489},
    number = {2},
    journal = {IEEE Transactions on Information Theory},
    publisher = {Institute of Electrical and Electronics Engineers (IEEE)},
    author = {Lloyd,  S.},
    year = {1982},
    month = {mar},
    pages = {129–137}
}

@inproceedings{macqueen_67,
  title={Some methods for classification and analysis of multivariate observations},
  author={MacQueen, James},
  booktitle={Proceedings of the Fifth Berkeley Symposium on Mathematical Statistics and Probability, Volume 1: Statistics},
  volume={5},
  pages={281--298},
  year={1967},
  organization={University of California press}
}

@inproceedings{newling_fleuret_16,
author = {Newling, James and Fleuret, Fran\c{c}ois},
title = {Nested mini-batch $k$-means},
year = {2016},
isbn = {9781510838819},
publisher = {Curran Associates Inc.},
address = {Red Hook, NY, USA},
booktitle = {Proceedings of the 30th International Conference on Neural Information Processing Systems},
pages = {1360–1368},
numpages = {9},
location = {Barcelona, Spain},
series = {NeurIPS'16},
url = {https://papers.nips.cc/paper_files/paper/2016/hash/8d317bdcf4aafcfc22149d77babee96d-Abstract.html}
}

@book{nielsen_chuang_10, 
    title     = "Quantum Computation and Quantum Information",
    author    = "Nielsen, Michael A and Chuang, Isaac L",
    publisher = "Cambridge University Press",
    month     =  {dec},
    year      =  {2010},
    address   = "Cambridge, England",
    isbn = {978-1-107-00217-3}
}

@incollection{phillips_17,
    title={Coresets and sketches},
    author={Phillips, Jeff M},
    booktitle={Handbook of discrete and computational geometry},
    year={2017}
}

@inproceedings{schwartzman_23,
    title={Mini-batch $k$-means terminates within $O(d/\epsilon)$ iterations}, 
    author={Gregory Schwartzman},
    year={2023},
    eprint={2304.00419},
    archivePrefix={arXiv},
    primaryClass={cs.LG},
    booktitle = {11th International Conference on Learning Representations (ICLR 2023)},
    isbn={9781713899259},
    url={https://arxiv.org/abs/2304.00419}, 
}

@inproceedings{sculley_10,
  series = {WWW ’10},
  title = {Web-scale k-means clustering},
  url = {http://dx.doi.org/10.1145/1772690.1772862},
  DOI = {10.1145/1772690.1772862},
  booktitle = {Proceedings of the 19th international conference on World wide web},
  publisher = {ACM},
  author = {Sculley,  D.},
  year = {2010},
  month = {apr},
  collection = {WWW ’10}
}

@article{scikit_11,
  title={Scikit-learn: {M}achine learning in {P}ython},
  author={Pedregosa, Fabian and Varoquaux, Ga{\"e}l and Gramfort, Alexandre and Michel, Vincent and Thirion, Bertrand and Grisel, Olivier and Blondel, Mathieu and Prettenhofer, Peter and Weiss, Ron and Dubourg, Vincent and others},
  journal={\JMLR{2011}},
  year={2011}
}

@InProceedings{so_mahajan_dasgupta_22,
  title = 	 { Convergence of online k-means },
  author =       {So, Geelon and Mahajan, Gaurav and Dasgupta, Sanjoy},
  booktitle = 	 {Proceedings of The 25th International Conference on Artificial Intelligence and Statistics},
  pages = 	 {8534--8569},
  year = 	 {2022},
  editor = 	 {Camps-Valls, Gustau and Ruiz, Francisco J. R. and Valera, Isabel},
  volume = 	 {151},
  series = 	 {Proceedings of Machine Learning Research},
  month = 	 {28--30 Mar},
  publisher =    {PMLR},
  url = 	 {https://proceedings.mlr.press/v151/so22a.html},
}

@InProceedings{tang_monteleoni_17,
  title = 	 {{Convergence rate of stochastic k-means}},
  author = 	 {Tang, Cheng and Monteleoni, Claire},
  booktitle = 	 {Proceedings of the 20th International Conference on Artificial Intelligence and Statistics},
  pages = 	 {1495--1503},
  year = 	 {2017},
  editor = 	 {Singh, Aarti and Zhu, Jerry},
  volume = 	 {54},
  series = 	 {Proceedings of Machine Learning Research},
  month = 	 {20--22 Apr},
  publisher =    {PMLR},
  url = 	 {https://proceedings.mlr.press/v54/tang17b.html},
}

@article{shukla_vedula_24,
  title = {An efficient quantum algorithm for preparation of uniform quantum superposition states},
  volume = {23},
  ISSN = {1573-1332},
  url = {http://dx.doi.org/10.1007/s11128-024-04258-4},
  DOI = {10.1007/s11128-024-04258-4},
  number = {2},
  journal = {Quantum Information Processing},
  publisher = {Springer Science and Business Media LLC},
  author = {Shukla,  Alok and Vedula,  Prakash},
  year = {2024},
  month = {jan} 
}

@inproceedings{grover_96,
  series = {STOC ’96},
  title = {A fast quantum mechanical algorithm for database search},
  url = {http://dx.doi.org/10.1145/237814.237866},
  DOI = {10.1145/237814.237866},
  booktitle = {Proceedings of the twenty-eighth annual ACM symposium on Theory of computing  - STOC ’96},
  publisher = {ACM Press},
  author = {Grover,  Lov K.},
  year = {1996},
  pages = {212–219},
  collection = {STOC ’96}
}

@article{giovannetti_lloyd_maccone_08,
  title = {Quantum Random Access Memory},
  volume = {100},
  ISSN = {1079-7114},
  url = {http://dx.doi.org/10.1103/PhysRevLett.100.160501},
  DOI = {10.1103/physrevlett.100.160501},
  number = {16},
  journal = {Physical Review Letters},
  publisher = {American Physical Society (APS)},
  author = {Giovannetti,  Vittorio and Lloyd,  Seth and Maccone,  Lorenzo},
  year = {2008},
  month = {apr} 
}

@article{ruizperez_garciaescartin_17,
  title = {Quantum arithmetic with the quantum Fourier transform},
  volume = {16},
  ISSN = {1573-1332},
  url = {http://dx.doi.org/10.1007/s11128-017-1603-1},
  DOI = {10.1007/s11128-017-1603-1},
  number = {6},
  journal = {Quantum Information Processing},
  publisher = {Springer Science and Business Media LLC},
  author = {Ruiz-Perez,  Lidia and Garcia-Escartin,  Juan Carlos},
  year = {2017},
  month = {apr} 
}

@article{thomsen_gluck_axelsen_10,
  title = {Reversible arithmetic logic unit for quantum arithmetic},
  volume = {43},
  ISSN = {1751-8121},
  url = {http://dx.doi.org/10.1088/1751-8113/43/38/382002},
  DOI = {10.1088/1751-8113/43/38/382002},
  number = {38},
  journal = {Journal of Physics A: Mathematical and Theoretical},
  publisher = {IOP Publishing},
  author = {Thomsen,  Michael Kirkedal and Gl\"{u}ck,  Robert and Axelsen,  Holger Bock},
  year = {2010},
  month = {aug},
  pages = {382002}
}

@article{wang_li_lee_deb_lim_chattopadhyay_25,
  title = {A comprehensive study of quantum arithmetic circuits},
  volume = {383},
  ISSN = {1471-2962},
  url = {http://dx.doi.org/10.1098/rsta.2023.0392},
  DOI = {10.1098/rsta.2023.0392},
  number = {2288},
  journal = {Philosophical Transactions of the Royal Society A: Mathematical,  Physical and Engineering Sciences},
  publisher = {The Royal Society},
  author = {Wang,  Siyi and Li,  Xiufan and Lee,  Wei Jie Bryan and Deb,  Suman and Lim,  Eugene and Chattopadhyay,  Anupam},
  year = {2025},
  month = {jan} 
}

@article{buhrman_dewolf_02,
    title = {Complexity measures and decision tree complexity: a survey},
    volume = {288},
    ISSN = {0304-3975},
    url = {http://dx.doi.org/10.1016/S0304-3975(01)00144-X},
    DOI = {10.1016/s0304-3975(01)00144-x},
    number = {1},
    journal = {Theoretical Computer Science},
    publisher = {Elsevier BV},
    author = {Buhrman,  Harry and de Wolf,  Ronald},
    year = {2002},
    month = {oct},
    pages = {21–43}
}

@article{aaronson_21,
  title = {Open Problems Related to Quantum Query Complexity},
  volume = {2},
  ISSN = {2643-6817},
  url = {http://dx.doi.org/10.1145/3488559},
  DOI = {10.1145/3488559},
  number = {4},
  journal = {ACM Transactions on Quantum Computing},
  publisher = {Association for Computing Machinery (ACM)},
  author = {Aaronson,  Scott},
  year = {2021},
  month = {dec},
  pages = {1–9}
}

@inproceedings{gilyen_su_low_wiebe_19,
  series = {STOC ’19},
  title = {Quantum singular value transformation and beyond: exponential improvements for quantum matrix arithmetics},
  url = {http://dx.doi.org/10.1145/3313276.3316366},
  DOI = {10.1145/3313276.3316366},
  booktitle = {Proceedings of the 51st Annual ACM SIGACT Symposium on Theory of Computing},
  publisher = {ACM},
  author = {Gilyén,  András and Su,  Yuan and Low,  Guang Hao and Wiebe,  Nathan},
  year = {2019},
  month = {jun},
  pages = {193–204},
  collection = {STOC ’19}
}

@inproceedings{xue_chen_li_jiang_23,
author = {Xue, Yecheng and Chen, Xiaoyu and Li, Tongyang and Jiang, Shaofeng H.-C.},
title = {Near-optimal quantum coreset construction algorithms for clustering},
year = {2023},
publisher = {JMLR.org},
booktitle = {Proceedings of the 40th International Conference on Machine Learning},
articleno = {1620},
numpages = {32},
location = {Honolulu, Hawaii, USA},
series = {ICML'23}
}

@article{tomesh_gokhale_21,
  title = {Coreset Clustering on Small Quantum Computers},
  volume = {10},
  ISSN = {2079-9292},
  url = {http://dx.doi.org/10.3390/electronics10141690},
  DOI = {10.3390/electronics10141690},
  number = {14},
  journal = {Electronics},
  publisher = {MDPI AG},
  author = {Tomesh,  Teague and Gokhale,  Pranav and Anschuetz,  Eric R. and Chong,  Frederic T.},
  year = {2021},
  month = {jul},
  pages = {1690}
}
\clearpage

\appendix

\section{Useful facts}
\label{sec:analysis}

\subsection{$k$-means cost function}

We use the following bounds for the cost function throughout.
\begin{lemma}[Local optimality]
\label{thm:local_opt}
The $k$-means steps are locally optimal in the sense that the following hold:
\begin{enumerate}[label=(\roman*)]
    \item \label{clm:sets-minimizers} 
    Given centers $(\vec{c}_1, \ldots, \vec{c}_k)$, define $C_j= \{ i \in[n] :  j = \argmin_{j'\in[k]}\|\vec{v}_i - \vec{c}_{j'} \|\}$.
    Then 
    \begin{equation*}
        \sum_{j\in [k]} \sum_{i\in C_j} \| \vec{v}_i - \vec{c}_j\|^2
        \leq 
        \sum_{j\in [k]} \sum_{i\in C_j'} \| \vec{v}_i - \vec{c}_j\|^2,\text{ where $(C_1', \ldots, C_k')$ is a partition of $[n]$}.
    \end{equation*}
    \item \label{clm:centers-minimizers} 
    Given a partition $C_1, \ldots, C_k$ of $[n]$, define $\vec{c}_j = |C_j|^{-1} \sum_{i\in C_j} \vec{v}_i$.
    Then
    \begin{equation*}
        \sum_{j\in [k]} \sum_{i\in C_j} \| \vec{v}_i - \vec{c}_j\|^2
        \leq 
        \sum_{j\in [k]} \sum_{i\in C_j} \| \vec{v}_i - \vec{c}_j'\|^2,
        \text{ where $\vec{c}_1', \ldots, \vec{c}_k'\in\R^d$}.
    \end{equation*}
\end{enumerate}
\end{lemma}

\begin{proof}
    \Cref{clm:sets-minimizers} follows from the fact that if $i\in C_j$, then $\|\vec{v}_i - \vec{c}_j\| \leq \|\vec{v}_i - \vec{c}_{j'}\|$, for any $j'$, and \cref{clm:centers-minimizers} follows from the fact that a function of the form $\sum_{i\in C} \| \vec{v}_i - \vec{c}\|^2$ is is minimized at $|C|^{-1} \sum_i \vec{v}_i$.
\end{proof}

As mentioned in the introduction, the parameter $\phi$ (see \cref{eqn:kc_L0}) appearing in our analysis is bounded above by the quantities $\bar{\eta}$ and $\eta$ appearing in past work.
The following result exactly quantifies how much smaller $\phi$ is than $\bar{\eta}$.

\begin{lemma}\label{thm:leqfrob}
For any $\cjzs$,
\begin{equation*}
\phi = \bar{\eta} - \frac{1}{n}\sum_{j\in [k]} |C_j^0|\|\vec{c}_j\|^2.
\end{equation*}
where $\bar{\eta} = n^{-1} \sum_{i\in[n]} \|\vec{v}_i\|^2 =  n^{-1}\|\vec{V}\|_\F^2$.
\end{lemma}
\begin{proof}
Observe that
\begin{align}
\phi
&= \frac{1}{n}\sum_{j\in[k]} \sum_{i\in C_j^0} \|\vec{v}_i - \vec{c}_j\|^2
\\&= \frac{1}{n}\sum_{j\in[k]} \bigg( \sum_{i\in C_j^0} \|\vec{v}_i\|^2  + \|\vec{c}_j\|^2 - 2\vec{v}_i^\T \vec{c}_j \bigg)
\\&= \frac{1}{n}\sum_{i\in[n]} \|\vec{v}_i\|^2
- \frac{1}{n} \sum_{j\in[k]}|C_j^0| \|\vec{c}_j\|^2.
\end{align}
where we have used that $\sum_{i\in C_j^0} \vec{v}_i = |C_j^0| \vec{c}_j$.
\end{proof}

\clearpage
\subsection{Approximate matrix multiplication}

\begin{lemma}[Approximate matrix multiplication]\label{thm:approx_mm}
Consider a matrix $\vec{A} \in \R^{n\times d}$ and vector $\vec{b}\in \R^d$.
Independently sample random indices $s_1, \ldots, s_b$ such that $\Pr(s_j=i) = p_i$ and define the estimator
\begin{equation*}
\widehat{\vec{y}} = \frac{1}{b} \sum_{i\in[b]} \frac{1}{p_{s_i}} \vec{A}_{:,s_i} \vec{b}_{s_i}.
\end{equation*}
Then $\EE[ \widehat{\vec{y}} ] =\vec{A}\vec{b}$ and
\begin{equation*}
\EE\big[\| \widehat{\vec{y}} \|^2\big]
= \frac{1}{b} \sum_{i\in[n]} \frac{1}{p_i} \|\vec{A}_{:,i}\|^2 | \vec{b}_i |^2.
\end{equation*}
\end{lemma}

\begin{proof}
The mean computation is straightforward; since $s_1, \ldots, s_b$ are identically distributed,
\begin{equation}
\EE\big[\widehat{\vec{y}}\big]
= \EE\Bigg[ \frac{1}{p_{s_1}} \vec{A}_{:,s_1} \vec{b}_{s_1} \Bigg]
= \sum_{i\in[n]} p_i \frac{1}{p_i} \vec{A}_{:,i} \vec{b}_i
= \sum_{i\in[n]} \vec{A}_{:,i} \vec{b}_i
= \vec{A}\vec{b}.
\end{equation}
We now compute the variance.
Using that $s_1, \ldots, s_b$ are independent and identically distributed, and that $\EE[p_{s_i}^{-1}\vec{A}_{:,s_i}\vec{b}_{s_i}] = \vec{A}\vec{b}$ so that cross-terms vanish,
\begin{align}\label{eqn:mv_var}
    \EE\big[\|\widehat{\vec{y}} \|^2\big]
    &= \EE\Bigg[\bigg\| \frac{1}{b}\sum_{i\in[b]} \Bigg(\frac{1}{p_{s_i}} \vec{A}_{:,s_i} \vec{b}_{s_i} \Bigg)\bigg\|^2\Bigg]
    \\&= \frac{1}{b}\EE\Bigg[ \bigg\| \frac{1}{p_{s_i}} \vec{A}_{:,s_i} \vec{b}_{s_i} \bigg\|^2\Bigg]
    \\&= 
    \frac{1}{b}\EE\Bigg[ \frac{1}{p_{s_i}^2} \| \vec{A}_{:,s_i} \|^2 | \vec{b}_{s_i} |^2 \Bigg]
    \\&= \frac{1}{b}\sum_{i\in[n]} \frac{1}{p_{i}} \| \vec{A}_{:,i} \|^2 | \vec{b}_{i} |^2.
\end{align}
This is the desired result.
\end{proof}

\subsection{Concentration inequalities}

We also state two standard concentration bounds which we have used.

\begin{imptheorem}[Markov inequality]\label{thm:markov}
Let $X$ be a non-negative random variable. 
Then, for any $\alpha > 0$,
\begin{equation*}
        \Pr\big( X > \alpha \big) \leq \frac{\EE[X]}{\alpha}.
\end{equation*}
\end{imptheorem}

\begin{imptheorem}[Chernoff bound]\label{thm:chernoff}
    Let $X_1, \ldots, X_b$ be independent copies of a Bernoulli random variable with success probability $p$ and let $X = X_1 + \cdots + X_n$.
    Then, for any $\delta\in(0,1)$
    \begin{equation*}
        \Pr\big(X \leq (1-\delta) bp \big) \leq \exp \Bigg( -\frac{\delta^2 bp}{2} \Bigg).
    \end{equation*}
\end{imptheorem}

\subsection{Median trick in high dimension}

We now describe how to efficiently turn an estimator which produces a $\varepsilon$-approximation to some high-dimensional quantity with constant probability into one which produces an $O(\varepsilon)$-approximation with arbitrary probability.

\begin{mdalg}[High dimensional median trick]~
    \label{alg:median_trick}
    \begin{mdframed}[style=alginner]~
    \begin{algorithmic}[1]
        \Require{Random vector $\vec{X}\in\R^d$, parameter $t \in \mathbb{N}$}
        \State Make independent copies $\vec{X}_1, \ldots, \vec{X}_t$ of $\vec{X}$
        \State For all $i,j\in[t]$, define $d(i,j) = \| \vec{X}_i - \vec{X}_j \|$
        \State For all $i\in[t]$, define $c(i) = \operatorname{median}(d(i,1),\ldots, d(i,t))$
        \State $i^* = \argmin_{i\in[t]} c(i)$
        \State $\vec{X}^* = \vec{X}_{i^*}$
        \Ensure{$\vec{X}^*$}
    \end{algorithmic}
    \end{mdframed}
\end{mdalg}

\begin{theorem}\label{thm:median_of_means}
    Let $\vec{X}\in\R^d$ be a random vector, and suppose there is some vector $\bm{\mu}$ such that\\ $\Pr(\| \bm{\mu}-\vec{X} \| > \varepsilon) < 1/3$.
    Then, for $t = O(\log(1/\delta))$ the output $\vec{X}^*$ of \cref{alg:median_trick} satisfies
    \[
    \Pr\big(\| \bm{\mu}-\vec{X}^* \| > 3\varepsilon\big) < \delta.
    \]
\end{theorem}

\begin{proof}

Let $G = \{ i : \| \bm{\mu} - \vec{X}_i \| \leq \varepsilon\}$.
Denoting $\Ind[A]$ to be the indicator function of the event $A$ (i.e., $\Ind[A](\omega) = 1$ if $\omega \in A$ and $0$ otherwise), we can write $|G| = \sum_{i = 1}^t \Ind[\| \bm{\mu} - \vec{X}_i \| \leq \varepsilon]$.
Since $\mathbb{E}[|G|] > 2 t/3$, it follows from \cref{thm:chernoff} that $\Pr(|G| \leq t/2)\leq \delta$ for $t = O(\log(1/\delta))$.
Therefore, it suffices to show that if $|G|> t/2$, then $\| \bm{\mu}-\vec{X}^* \| < 3\varepsilon$.

Suppose $|G| > t/2$.
By the triangle inequality, for each $i,j\in[t]$,
\begin{equation}
    d(i,j) = \|\vec{X}_i - \vec{X}_j\|
    \leq \| \bm{\mu} - \vec{X}_i \| + \| \bm{\mu} - \vec{X}_j \|.
\end{equation}
Therefore, for each $i\in G$
\begin{equation}
    \big|\{j\in [t] : d(i,j) \leq 2\varepsilon\}\big|
    > t/2,
\end{equation}
and hence, 
\[
c(i) = \operatorname{median}\big(d(i,1),\ldots, d(i,t)\big) \leq 2\varepsilon.
\]

It follows that $c(i^*) = \min_i c(i) \leq 2\varepsilon$, and since $c(i^*) = \operatorname{median}(d(i^*,1),\ldots, d(i^*,t))$, we have
\begin{equation}
    \big|\{j\in [t] : d(i^*,j) \leq 2\varepsilon\}\big|
    \geq t/2.
\end{equation}
But also $|G|>t/2$. 
So, by the pigeonhole principle, there is at lest one index $j^*\in G$ for which 
\begin{equation}
    d(i^*,j^*) \leq 2\varepsilon.
\end{equation}
Therefore, by the triangle inequality
\begin{equation}
    \| \bm{\mu} - \vec{X}_{i^*} \|
    \leq \| \bm{\mu} - \vec{X}_{j^*} \| + \| \vec{X}_{i^*} - \vec{X}_{j^*} \|
    \leq \varepsilon + 2\varepsilon = 3\varepsilon,
\end{equation}
as desired.
\end{proof}

\section{Bounds for the $k$-means cost function}

The following shows that \cref{thm:main} also implies the cost  the final cluster centers output by \cref{alg:minibatch_kmeans} is not much larger than the cost of the initial clusters.
\begin{corollary}\label{thm:monotone_cor}
Suppose
\begin{equation*}
    b \geq k_C\cdot \max\Bigg\{\frac{4}{\varepsilon\delta}, 8\log\bigg(\frac{k}{\delta}\bigg)\Bigg\}.
\end{equation*}
Then, with probability at least $1-\delta$, the output $\cjhs$ of \cref{alg:minibatch_kmeans} satisfies 
\begin{equation*}
    \mathcal{L}\cjhs
    \leq (1+\varepsilon) \phi
    \leq (1+\varepsilon) \mathcal{L}\cjzs.
\end{equation*}
\end{corollary}

The following (standard) lemma provides a quantitative bound on how the cost function behaves for centers near the $k$-means centers.
\begin{lemma}\label{thm:cost_smoothness}
For any cluster centers $\cjhs$, 
\begin{equation*}
\mathcal{L}\cjhs
\leq  \phi
+ \frac{1}{n}\sum_{j\in [k]} |C_j^0|  \| \vec{c}_j - \widehat{\vec{c}}_j \|^2.
\end{equation*}
\end{lemma}

\begin{proof}
Define $\widehat{C}_j = \{ i \in[n] :  j = \argmin_{j'\in[k]}\|\vec{v}_i - \widehat{\vec{c}}_{j'} \|\}$.
Then
\begin{equation}
\label{eqn:Lhatc_bd}
\mathcal{L}\cjhs
=
\frac{1}{n}\sum_{j\in[k]} \sum_{i\in \widehat{C}_j}  \| \vec{v}_i - \widehat{\vec{c}}_j\|^2
\leq 
\frac{1}{n}\sum_{j\in[k]} \sum_{i\in C_j^0}  \| \vec{v}_i - \widehat{\vec{c}}_j\|^2,
\end{equation}
where the last inequality follows from the optimality of the partition $(\widehat{C}_1, \ldots, \widehat{C}_k)$ for centers $\cjhs$.

Next, observe that for any $\widehat{\vec{c}}_j \in \R^d$,
\begin{equation}
    \| \vec{v}_i - \widehat{\vec{c}}_j \|^2
    = \| \vec{v}_i - \vec{c}_j + \vec{c}_j - \widehat{\vec{c}}_j \|^2
    = \| \vec{v}_i - \vec{c}_j \|^2 + \| \vec{c}_j - \widehat{\vec{c}}_i \|^2 + 2 (\vec{v}_i - \vec{c}_j)^\T (\vec{c}_j - \widehat{\vec{c}}_j).
\end{equation}
Since, by definition $\vec{c}_j = |C_j^0|^{-1} \sum_{i\in C_j^0} \vec{v}_i$, by linearity 
\begin{equation}
\sum_{i\in C_j^0} 2 (\vec{v}_i - \vec{c}_j)^\T (\vec{c}_j - \widehat{\vec{c}}_j)
= \vec{0}.
\end{equation}
Therefore, since $(\vec{c}_1, \ldots, \vec{c}_k)$ are the optimal cluster center with respect to the partition $(C_1^0, \ldots, C_k^0)$,
\begin{align}
\frac{1}{n}\sum_{j\in[k]}\sum_{i\in C_j^0}  \| \vec{v}_i - \widehat{\vec{c}}_j \|^2 
&= \sum_{j\in[k]}\sum_{i\in C_j^0} \Bigg(\| \vec{v}_i - \vec{c}_j\|^2 + \| \vec{c}_j - \widehat{\vec{c}}_j \|^2 \Bigg)
\\&= \phi + \frac{1}{n}\sum_{j\in[k]}  |C_j^0| \| \vec{c}_j - \widehat{\vec{c}}_j \|^2.
\label{eqn:vihatc_bd}
\end{align}
Plugging \eqref{eqn:vihatc_bd} into \eqref{eqn:Lhatc_bd} gives the result.
\end{proof}

\begin{proof}[Proof of \cref{thm:monotone_cor}]
    The first result follows by plugging \cref{thm:main} into \cref{thm:cost_smoothness} and scaling $\varepsilon$ appropriately.
    The second is due to the optimality of $\cjs$ for the clusters $(C_1^0, \ldots, C_k^0)$.
\end{proof}

\section{Damping}
\label{sec:damping}

Both the mini-batch and quantum algorithms described above perform an approximate $k$-means step. 
While the variance in these updates can be controlled by decreasing the accuracy parameter $\varepsilon$, an alternative approach to decreasing the variance is to damp the iterates by interpolating between the initial cluster center and the approximate update. 
In particular, given damping coefficients $\alpha_1, \ldots, \alpha_k \in (0,1]$, we might aim to emulate a damped $k$-means update 
\begin{equation}
    \vec{c}_j^\alpha = (1-\alpha_j) \vec{c}_j^0 + \alpha_j \vec{c}_j.
\end{equation}
In the case that $\alpha_j = 1$ then this recovers the standard $k$-means update.

\subsection{Damped mini-batch algorithm and bound}
 
The damped $k$-means update naturally gives rise to a generic damped mininbatch $k$-means algorithm which we describe in \cref{alg:damped_minibatch_kmeans}.
This variant is closely related to the \href{https://scikit-learn.org/stable/modules/generated/sklearn.cluster.MiniBatchKMeans.html}{\texttt{MiniBatchKMeans}} method implemented in \texttt{scikit-learn} \cite{scikit_11}.

\begin{mdalg}[Generic damped mini-batch $k$-means (one iteration)]~
    \label{alg:damped_minibatch_kmeans}
    \begin{mdframed}[style=alginner]~
    \begin{algorithmic}[1]
        \Require{Initial centers $\vec c_1^0, \dots, \vec c_k^0 \in \R^d$, data $\vec{v}_1, \ldots, \vec{v}_n \in \R^d$, damping coefficients $\alpha_1, \ldots, \alpha_k\in[0,1]$}
        \State Sample $b$ indices $B = (s_1, \dots, s_b) \in [n]^b$, each independently such that $\Pr(s_\ell = i) = 1/n$.
        \For {$j \in [k]$}
        \State $\widehat{C}_j^0 = \{ i\in B : j = \argmin_{j'\in[k]} \| \vec{v}_i - \vec{c}_{j'}^0 \| \}$
        \State $\displaystyle \widehat{\vec{c}}_j =  \frac{1}{|\widehat{C}_j^0|} \sum_{i\in \widehat{C}_j^0} \vec v_{i}$
        \State $\widehat{\vec{c}}_j^\alpha = (1-\alpha_j) \vec{c}_j^0 + \alpha_j\widehat{\vec{c}}_j$
        \EndFor
        \vspace{-1em}\Ensure{updated centers $\widehat{\vec{c}}_1^\alpha, \dots, \widehat{\vec{c}}_k^\alpha \in \R^d$}
    \end{algorithmic}
    \end{mdframed}
\end{mdalg}

Our main result for \cref{alg:damped_minibatch_kmeans} is a generalization of \cref{thm:monotone_cor}.

\begin{corollary}\label{thm:minibatch_monotone_cor}
Let $\alpha_{\textup{max}} = \max_j \alpha_j$ and $\alpha_{\textup{min}} = \min_j\alpha_j$, and suppose
\begin{equation*}
    b \geq k_C\cdot \max\bigg\{\frac{40}{\varepsilon}, 8\log(20k)\bigg\}.
\end{equation*}
Then, with probability at least $9/10$, the output $(\widehat{\vec{c}}_1^\alpha, \ldots, \widehat{\vec{c}}_k^\alpha)$ of \cref{alg:damped_minibatch_kmeans} satisfies 
 \begin{equation*}
    \mathcal{L}\big(\widehat{\vec{c}}_1^\alpha, \ldots, \widehat{\vec{c}}_k^\alpha\big)
    \leq \Bigg( (1-\alpha_{\textup{min}}) + (\alpha_{\textup{max}}) \sqrt{1+\varepsilon} \Bigg)^2
    \phi.
    \end{equation*}
\end{corollary}

In particular, if $\alpha_j = \alpha$ are constant, then
\begin{equation}
    \mathcal{L}\big(\widehat{\vec{c}}_1^\alpha, \ldots, \widehat{\vec{c}}_k^\alpha\big)
    \leq (1+\varepsilon )\phi
    \leq (1+\varepsilon) \mathcal{L}\cjzs,
\end{equation}
which also generalizes \cref{thm:monotone_cor}.

\subsection{Proofs}
We begin by proving a generalization of \cref{thm:cost_smoothness}.
\begin{lemma}\label{thm:minibatch_cost_smoothness}
Let $\alpha_{\textup{max}} = \max_j \alpha_j$ and $\alpha_{\textup{min}} = \min_j\alpha_j$.
For data points $\vec{v}_1, \ldots, \vec{v}_n \in \R^d$ and given centers $\vec{c}_1^0, \ldots, \vec{c}_k^0 \in \R^d$, define $C_j^0= \{ i \in[n] :  j = \argmin_{j'\in[k]}\|\vec{v}_i - \vec{c}_{j'}^0 \|\}$. For any cluster centers defined as~ $\widehat{\vec{c}}_j^\alpha = (1-\alpha_j) \vec{c}_j^0 + \alpha_j \widehat{\vec{c}}_j$,  we have
\begin{equation*}
\mathcal{L}\big(\widehat{\vec{c}}_1^\alpha, \ldots, \widehat{\vec{c}}_k^\alpha\big)^{1/2}
\leq 
(1-\alpha_{\textup{min}}) \phi^{1/2}
+ (\alpha_{\textup{max}})\Bigg(\phi + \frac{1}{n}\sum_{j\in[k]}|C_j^0| \| \vec{c}_j - \widehat{\vec{c}}_j \|^2 \Bigg)^{1/2}.
\end{equation*}
\end{lemma}

\begin{proof}
    From \cref{eqn:cost} for centers $\widehat{\vec{c}}_1^\alpha, \ldots, \widehat{\vec{c}}_k^\alpha\in \R^d$, and induced clusters $\widehat{C}_1^\alpha, \ldots, \widehat{C}_k^\alpha$ we have
    \begin{equation}
    \mathcal{L}\big(\widehat{\vec{c}}_1^\alpha, \ldots, \widehat{\vec{c}}_k^\alpha\big)
    = \frac{1}{n}\sum_{j\in[k]} \sum_{i\in \widehat{C}_j^\alpha} \| \vec{v}_i - \widehat{\vec{c}}_j^\alpha\|^2
    \leq \frac{1}{n}\sum_{j\in[k]} \sum_{i\in C_j^0} \| \vec{v}_i - \widehat{\vec{c}}_j^\alpha\|^2,
    \end{equation}
    where for the inequality we use \cref{thm:local_opt}~\cref{clm:sets-minimizers}. Then, from the definition of $\widehat{\vec{c}}_j^\alpha$ we have
    \begin{align}
        \sum_{j\in[k]} \sum_{i\in C_j^0} \| \vec{v}_i - \widehat{\vec{c}}_j^\alpha\|^2
        &= 
        \sum_{j\in[k]} \sum_{i\in C_j^0} \| (1-\alpha_j) (\vec{v}_i - \vec{c}_j^0) + \alpha_j(\vec{v}_i - \widehat{\vec{c}}_j)\|^2\\
        &=
        \sum_{j\in[k]} \sum_{i\in C_j^0} \Bigg( (1-\alpha_j)^2 \| \vec{v}_i - \vec{c}_j^0 \|^2 + (\alpha_j)^2 \|\vec{v}_i - \widehat{\vec{c}}_j\|^2 + 2 \alpha_j (1-\alpha_j) (\vec{v}_i - \widehat{\vec{c}}_j)^\T (\vec{v}_i - \vec{c}_j^0) \Bigg).\label{eqn:cjalpha-expanded}
    \end{align}
    Define
    \begin{equation}\label{eqn:aux_vars}
        A_j = \sum_{i\in C_j^0} \| \vec{v}_i - \vec{c}_j^0 \|^2
        ,\qquad
        B_j = \sum_{i\in C_j^0} \|\vec{v}_i - \widehat{\vec{c}}_j\|^2.
    \end{equation}
    Also observe that by Cauchy--Schwarz,\footnote{Specifically, we use that $| \sum_i \vec{x}_i^\T \vec{y}_i |^2
    = \tr(\vec{X}^\T\vec{Y})^2
    \leq ( \sum_i\| \vec{x}_i \|^2 )(\sum_i \|\vec{y}_i\|^2)$.}
    \begin{align}\label{eqn:CS_applied}
        \Bigg(\sum_{j\in[k]} \sum_{i\in C_j^0} \alpha_j (1-\alpha_j) (\vec{v}_i - \widehat{\vec{c}}_j)^\T (\vec{v}_i - \vec{c}_j^0)  \Bigg)^2
        \leq  \Bigg(\sum_{j\in[k]} \sum_{i\in C_j^0} (\alpha_j)^2 \| \vec{v}_i - \widehat{\vec{c}}_j \|^2  \Bigg)\Bigg(\sum_{j\in[k]} \sum_{i\in C_j^0}  (1-\alpha_j)^2 \| \vec{v}_i - \vec{c}_j^0\|^2  \Bigg).
    \end{align}
    Then, using \cref{eqn:aux_vars} and \cref{eqn:CS_applied} in \cref{eqn:cjalpha-expanded}, and noting that $\alpha_j$ is invariant under the sum over $i\in C_j^0$, we have
    \begin{align}
        \sum_{j\in[k]} \sum_{i\in C_j^0} \| \vec{v}_i - \widehat{\vec{c}}_j^\alpha\|^2 
        &\leq \sum_{j\in [k]} (1-\alpha_j)^2 A_j + \sum_{j\in [k]}(\alpha_j)^2 B_j + 2\sqrt{ \Bigg( \sum_{j\in[k]} (1-\alpha_j)^2 A_j \Bigg)  \Bigg(\sum_{j\in[k]} (\alpha_j)^2B_j \Bigg)}
        \\&= \Bigg( \Bigg(\sum_{j\in [k]} (1-\alpha_j)^2 A_j\Bigg)^{1/2} + \Bigg(\sum_{j\in[k]} (\alpha_j)^2B_j\Bigg)^{1/2} \Bigg)^2
        \\&\leq \Bigg(  (1-\alpha_{\textup{min}}) \Bigg(\sum_{j\in [k]} A_j\Bigg)^{1/2} + \alpha_{\textup{max}}\Bigg(\sum_{j\in[k]} B_j\Bigg)^{1/2} \Bigg)^2.\label{eqn:cost-chat-j-alpha}
    \end{align}
    By definition,
    \begin{equation}
        \frac{1}{n}\sum_{j\in[k]} A_j = \frac{1}{n} \sum_{j\in[k]} \sum_{i\in C_j^0} \| \vec{v}_i - \vec{c}_j^0 \|^2
        = \phi,
    \end{equation}
    and as in the proof of \cref{thm:cost_smoothness}, 
    \begin{equation}
        \frac{1}{n}\sum_{j\in[k]}B_j
        = 
        \frac{1}{n}\sum_{j\in[k]}\sum_{i\in C_j^0}  \| \vec{v}_i - \widehat{\vec{c}}_j \|^2 
        = \phi + \frac{1}{n}\sum_{j\in[k]}  |C_j^0| \| \vec{c}_j - \widehat{\vec{c}}_j \|^2.
    \end{equation}
    Putting everything back in \cref{eqn:cost-chat-j-alpha} and taking a square-root gives us the lemma.
\end{proof}

\begin{proof}[Proof of \cref{thm:minibatch_monotone_cor}.]
This follows from plugging \cref{thm:main} into \cref{thm:minibatch_cost_smoothness}.
\end{proof}

\section{Additional Details on Past Work}
\label{sec:past_quantum_only}

\subsection{q-means: A Quantum Algorithm for Unsupervised Machine Learning \cite{kerenidis_landman_luongo_prakash_19}}\label{sec:past:Ker19}

$q$-means is a quantum algorithm which aims to mimic the $k$-means algorithm over one step.
The algorithm assumes superposition query-access to the dataset containing $n$ data points via quantum-random-access-memory (QRAM) in time logarithmic in the size of the dataset, where each data point is encoded into the \emph{amplitude-encoding} quantum state using the KP-tree data structure \cite{kerenidis_prakash_17}. 
Given the initial centroids $\vec{c}_1^0, \ldots, \vec{c}_k^0$, the $q$-means algorithm produces centroids $\widehat{\vec{c}}_1, \ldots, \widehat{\vec{c}}_k$ that are close to the classical $k$-means algorithm (see \cite[\S1.3]{kerenidis_landman_luongo_prakash_19}) 
in the sense that the approximate centers satisfy
\begin{equation}\label{eqn:delta_kmeans1}
    \forall j\in [k]:\quad \bigg\| \widehat{\vec{c}}_j - \frac{1}{|\breve{C}_j^0|}\sum_{i\in \breve{C}_j^0} \vec{v}_i\bigg\| \leq \varepsilon,
\end{equation}
where $(\breve{C}_1^0, \ldots, \breve{C}_k^0)$ is some partition of $[n]$ for which 
\begin{equation}\label{eqn:delta_kmeans2}
    \forall j\in[k],~i \in \breve{C}_j^0:\quad 
    \| \vec{v}_i - \vec{c}_j^0\|^2 \leq \min_{j'\in [k]}\| \vec{v}_i - \vec{c}_{j'}^0\|^2 + \tau.
\end{equation}
That is, the approximate centers are nearly the mean of the right partition of the dataset.
The parameter $\tau$ measures the nearness to the $k$-means algorithm; if $\tau = 0$, then $\widehat{\vec{c}}_1, \ldots, \widehat{\vec{c}}_k$ must match the centroids produced by the $k$-means algorithm.

Under the assumptions that $\min_{i\in[n]} \|\vec{v}_i\| \geq 1$ and that the initial cluster sizes are roughly balanced ($|C_j^0| = \Theta(n/k)$), \cite[Theorem 5.1]{kerenidis_landman_luongo_prakash_19} asserts that the $q$-means algorithm produces iterates satisfying \cref{eqn:delta_kmeans1} and \cref{eqn:delta_kmeans2} (with $\tau=\varepsilon$\footnote{Note that the units on $\varepsilon$ are inconsistent between \cref{eqn:delta_kmeans1,eqn:delta_kmeans2} so the bound is not scale invariant.
This is implicitly mitigated by the assumption that $\|\vec{v}_i\|\geq 1$ for each $i\in[n]$.}) in 
\begin{equation}
\label{eqn:qmeans_runtime}
    \tilde{O}\Bigg(\frac{kd }{\varepsilon^2}\eta \kappa \bigg(\mu + \frac{k\eta}{\varepsilon}\bigg) + \frac{k^2\eta^{3/2}\kappa\mu}{\varepsilon^2}\Bigg) ~~ \text{time per-iteration,}
\end{equation}
where $\tilde{O}(\,\cdot\,)$ hides poly-logarithmic dependencies, $\eta \coloneq \max_i\|\vec{v}_i\|^2$, $\kappa$ is the condition number of data matrix $\vec{V} = [\vec{v}_1,\ldots,\vec{v}_n]^\T$, and $\mu \leq \|\vec{V}\|_\F$ is a data-dependent quantity which arises from the use of black-box quantum linear algebra subroutines within the $q$-means algorithm.  
The main primitives used in the $q$-means algorithm include quantum distance estimation, minimum distance finding, and quantum linear algebra subroutines to approximate a matrix-vector product. 

\subsection{Do You Know What $q$-means?  \texorpdfstring{\cite{doriguello_luongo_tang_25, cornelissen_doriguello_luongo_tang_25}}{}}\label{sec:past:DLT25}

The authors in \cite{doriguello_luongo_tang_25} introduce an ``improved'' version of the $q$-means algorithm which avoids the need to use quantum linear algebra subroutines required by \cite{kerenidis_landman_luongo_prakash_19}, while still maintaining a poly-logarithmic dependence in $n$.
As such, the dependence on most parameters is improved. 
Additionally, they introduce a ``quantum-inspired'' classical algorithm that maintains the poly-logarithmic dependence in $n$, although with a quadratically worse bound in data-dependent parameters.
Additionally, they strengthen the approximation guarantee over \cite{kerenidis_landman_luongo_prakash_19} by producing clustering centers $\widehat{\vec{c}}_1, \ldots, \widehat{\vec{c}}_k$ satisfying
\begin{equation}\label{eqn:delta_kmeans3}
    \forall j\in [k]:\quad \| \widehat{\vec{c}}_j - \vec{c}_j \| \leq \varepsilon,
\end{equation}
where $\vec{c}_1, \ldots, \vec{c}_k$ are the centroids output by the $k$-means algorithm.
This guarantee corresponds to that of \cite{kerenidis_landman_luongo_prakash_19} with $\varepsilon$ in \cref{eqn:delta_kmeans2} set to zero.

The quantum algorithm of \cite{doriguello_luongo_tang_25} assumes query-access to the data-set via QRAM which is able to prepare row-norm weighted superpositions over \emph{binary encoding} of the data points.
Then, under the same assumptions $\min_{i\in[n]} \|\vec{v}_i\| \geq 1$ and $|C_j^0| = \Theta(n/k)$, \cite[Theorem 11]{doriguello_luongo_tang_25} asserts that the quantum algorithm runs in 
\begin{equation}
\label{eqn:qmeans2_runtime}
    \tilde{O}\Bigg( \frac{k^{2}d\sqrt{\bar{\eta}}}{\varepsilon} (\sqrt{k} + \sqrt{d}) \Bigg)
    ~~ \text{time per-iteration,\footnotemark}
\end{equation}
where $\bar{\eta} = n^{-1}\sum_{i\in[n]} \|\vec{v}_i\|^2 = n^{-1}\|\vec{V}\|_\F^2 \leq \eta$.
\footnotetext{\label{foot:DLT25}To the best of our knowledge, only the query complexity (and not time-complexity) of the mean estimation procedure used by \cite{doriguello_luongo_tang_25} was analyzed. Hence, the bound they state is more analogous to a bound on the QRAM query complexity (assuming $\tilde{O}(d)$ time/query) than a true runtime guarantee.}
For a fair comparison with the results of our work, we note that their algorithm requires 
\begin{equation}
\label{eqn:qmeans2_querycompelexity}
    \tilde{O}\Bigg( \frac{k^{3/2}\sqrt{\bar{\eta}}}{\varepsilon} (\sqrt{k} + \sqrt{d}) \Bigg)
    ~~ \text{QRAM queries per-iteration},
\end{equation}
which we obtained by counting the number of QRAM calls in~\cite[Theorem~11, Lemma~12]{doriguello_luongo_tang_25}.

\begin{figure}
\begin{mdalg}[Algorithm 1 in \cite{doriguello_luongo_tang_25,cornelissen_doriguello_luongo_tang_25} (one iteration)]~
    \label{alg:DLT}
    \begin{mdframed}[style=alginner]~
    \begin{algorithmic}[1]
        \Require{Data $\vec{v}_1, \ldots, \vec{v}_n \in \R^d$, initial centers $\vec c_1^0, \dots, \vec c_k^0 \in \R^d$}
        \State Sample $a$ indices $A = \{t_1, \dots, t_a\} \in [n]^a$, each independently such that $\Pr(t_\ell = i) = 1/n$.
        \State Sample $b$ indices $B = \{s_1, \dots, s_b\} \in [n]^b$, each independently such that 
        \[\Pr(s_\ell = i) \propto \begin{cases}
            \|\vec{v}_i\|^2 & \text{\cite{doriguello_luongo_tang_25}} \\
            \|\vec{v}_i\| & \text{\cite{cornelissen_doriguello_luongo_tang_25}} \\
        \end{cases}.\]
        \For {$j \in [k]$}
        \State $\widetilde{C}_j^0 = \Big\{ i\in A : j = \argmin_{j'\in[k]} \| \vec{v}_i - \vec{c}_{j'}^0 \| \Big\}$
        \State $\widehat{C}_j^0 = \Big\{ i\in B : j = \argmin_{j'\in[k]} \| \vec{v}_i - \vec{c}_{j'}^0 \| \Big\}$
        \State 
        $\displaystyle\widehat{\vec{c}}_j =  \frac{a}{n|\widetilde{C}_j^0|} \sum_{i\in \widehat{C}_j^0}  \frac{1}{b \Pr(s_\ell = i)} \vec v_{i}$
        \EndFor
        \vspace{-1em}\Ensure{Updated centers $\widehat{\vec{c}}_1, \dots, \widehat{\vec{c}}_k \in \R^d$}
    \end{algorithmic}
    \end{mdframed}
\end{mdalg}
\end{figure}

The quantum-inspired  classical algorithm of \cite{doriguello_luongo_tang_25} mimics the quantum algorithm by assuming access to the sample-and-query (SQ) access model
to draw a subset of the data, where data points are drawn proportional to the squared row-norms $\|\vec{v}_i\|^2$ \cite{doriguello_luongo_tang_25}.
Under the same balanced clusters assumption, \cite[Theorem 11]{doriguello_luongo_tang_25} asserts that the dequantized algorithm uses
\begin{equation}
\label{eqn:eps_kmeans_runtime}
    O\Bigg( \frac{k^2\log(k)\bar{\eta}}{\varepsilon^2}  \Bigg) 
~~ \text{samples per-iteration.}
\end{equation}
The KP-tree data-structure \cite{kerenidis_prakash_17} enables SQ access by sampling an index $i \in [n]$ 
proportional to the squared row-norms in $\tilde{O}(1)$ time, and reading a data point in $\tilde{O}(d)$ time respectively. Thus the total run-time cost of the algorithm is $\tilde{O}(\textup{\#samples} \times kd)$ time per iteration.

We emphasize that both the quantum and quantum-inspired algorithm of \cite{doriguello_luongo_tang_25} make use of importance sampling: data points $\vec{v}_i$, for which $\|\vec{v}_i\|$ is large, are deemed to be more important in estimating the centroids.
As such, as we discuss in \cref{sec:symmetry}, they incur a dependence on the $\bar{\eta}$ or $\eta$, which may be large, even on intuitively easy to cluster data.
In contrast, our algorithm satisfy bounds in terms of $\phi\leq \bar{\eta}$, which is small when the data is well-clusterable.

\subsubsection{Updated version \texorpdfstring{\cite{cornelissen_doriguello_luongo_tang_25}}{}}

After the first version of the present manuscript was posted, a new version \cite{cornelissen_doriguello_luongo_tang_25} of \cite{doriguello_luongo_tang_25} was posted.
This new version impoves over past versions.
In particular, \cite[Theorem 14]{cornelissen_doriguello_luongo_tang_25} asserts that the new quantum algorithm runs in 
\begin{equation}
\label{eqn:qmeans2_runtimev3}
    \tilde{O}\Bigg( \frac{k^{3/2}}{\varepsilon} \Bigg(\sqrt{k} \frac{\|\vec{V}\|}{\sqrt{n}} + \frac{\sqrt{d}}{n} \sum_{i=1}^{n} \|\vec{v}_i\| \Bigg) \Bigg)
    ~~ \text{QRAM queries per-iteration,}
\end{equation}
where $\hat{\eta} = ( n^{-1}\sum_{i\in[n]} \|\vec{v}_i\|)^2 + n^{-1}\|\vec{V}\|^2$.
Likewise, \cite[Theorem 11]{cornelissen_doriguello_luongo_tang_25} asserts that the new dequantized algorithm uses
\begin{equation}
\label{eqn:eps_kmeans_runtimev3}
    O\Bigg( \frac{k^2\log(k)\hat{\eta}}{\varepsilon^2}  \Bigg) 
~~ \text{samples per-iteration.}
\end{equation}
In addition, algorithms and bounds for the case where the cluster assignment is noisy (i.e., $\tau > 0$ in \cref{eqn:delta_kmeans2}) are provided.

Let us understand the relation between the bounds \cref{eqn:eps_kmeans_runtime,eqn:eps_kmeans_runtimev3} for the dequantized algorithms from \cite{doriguello_luongo_tang_25,cornelissen_doriguello_luongo_tang_25}.
Since $\|\vec{V}\|_\F^2= \sum_{i=1}^{n} \|\vec{v}_i\|^2$, it holds that 
\begin{equation}
\frac{\| \vec{V}\|^2}{n} 
\leq \frac{\| \vec{V}\|_\F^2}{n}
\leq \frac{d\| \vec{V}\|^2}{n}
,\qquad
\left(\frac{1}{n}\sum_{i\in[n]} \|\vec{v}_i\|\right)^2
\leq \frac{\|\vec{V}\|_\F^2}{n}
\leq n\left(\frac{1}{n}\sum_{i\in[n]} \|\vec{v}_i\|\right)^2,
\end{equation}
and hence
\begin{equation}
\hat{\eta} \leq 2 \bar{\eta}
\leq 2 d \hat{\eta}.
\end{equation}
This implies \cref{eqn:eps_kmeans_runtimev3} is always better than (up to constants) \cref{eqn:eps_kmeans_runtime}, and possibly better by a factor of up to $d$.

We can also compare \cref{eqn:eps_kmeans_runtimev3} with \cref{thm:main}.
Using \cref{thm:leqfrob} we have that
\begin{equation}
    \phi = \bar{\eta} - \frac{1}{n}\sum_{j\in [k]} |C_j^0|\|\vec{c}_j\|^2
    \leq \bar{\eta} \leq d\hat{\eta}.
\end{equation}
Thus, our bound \cref{thm:main} is at most a factor of $d$ worse than \cref{eqn:eps_kmeans_runtimev3}.
For this bound to be attained ($\phi \approx d \hat{\eta}$), we require that all singular values of $\vec{V}$ are similar, that $\vec{V}$ has a few large rows, but most are small, and that the centers produced by $k$-means are near zero. 
All three conditions can be ensured if e.g., $\vec{V}^\T = [\vec{I}, -\vec{I}, \vec{0}]$ and $k=1$.
However, it is unclear that $k$-means clustering is appropriate for such data.
On the other hand, it is clear that there are settings for which $\phi \ll \hat{\eta}$; see \cref{sec:symmetry} for discussion.
As such, neither algorithm always does better than the other algorithm.

\end{document}